%% file: 0-qbu2po-main.tex
\documentclass{llncs}
\newif\ifanonoymous
\anonoymousfalse

\usepackage[T1]{fontenc}
\usepackage[utf8]{inputenc}
\usepackage{amsmath}
\usepackage{amssymb}
\usepackage{mathtools}
\usepackage{graphicx}
\usepackage{xcolor}
\usepackage{enumitem}
\usepackage{algorithm}
\usepackage{algorithmic}
\usepackage{tikz}
\usepackage{quantikz}
\usetikzlibrary{calc,fit,positioning,arrows.meta}
\definecolor{figbg}{RGB}{246,246,255}
\definecolor{figborder}{RGB}{176,179,226}
\definecolor{gateborder}{RGB}{90,94,150}

\usepackage[bookmarks=false,hidelinks,pagebackref]{hyperref}
\renewcommand*{\backref}[1]{}

\renewcommand*{\backrefalt}[4]{%
  \ifcase #1
    \relax
  \or
    \space(Cited on p.~#2.)%
  \else
    \space(Cited on pp.~#2.)%
  \fi
}
\usepackage[capitalise,nameinlink]{cleveref}   
\newcommand{\eps}{\varepsilon}
\newcommand{\negl}{\mathsf{negl}}
\newcommand{\poly}{\mathrm{poly}}

\newcommand{\E}{\mathop{\mathbb{E}}}

\newcommand{\X}{\mathcal X}
\newcommand{\Y}{\mathcal Y}
\newcommand{\K}{\mathcal K}
\newcommand{\Prb}{\mathop{\rm Pr}}
\newcommand{\PO}{\textsf{PO}}
\newcommand{\BU}{\textsf{BU}}
\newcommand{\dBU}{\partial\textsf{BU}}
\newcommand{\secpar}{\lambda}
\newcommand{\norm}[1]{\left\lVert#1\right\rVert}

\let\oldproof\proof
\let\oldendproof\endproof

\renewenvironment{proof}[1][]{%
  \if\relax\detokenize{#1}\relax
    \oldproof
  \else
    \oldproof[#1]
  \fi
}{%
  \unskip\nobreak\hfill\qed
  \oldendproof
}
\begin{document}

\title{
Blind Unforgeability $\Longrightarrow$ Plus-One Unforgeability
}
\ifanonoymous
  \author{}
  \authorrunning{Anonymous Submission \#\,\thanks{Anonymous submission to Eurocrypt 2027.}}
  \institute{}
\else
\author{
  Mehil Agarwal \and
  Calliope S. Reimann \and
  Fang Song
}
\authorrunning{M. Agarwal, C. Reimann, and F. Song}
\institute{
  Department of Computer Science, Portland State University, USA\\
  \email{\{mehil,creimann,fang.song\}@pdx.edu}
}
\fi

\maketitle
\thispagestyle{plain}   

\begin{abstract} 
We prove that quantum blind-unforgeability implies quantum plus-one unforgeability, settling an open question in the literature. Since plus-one unforgeability is known not to imply blind-unforgeability, our result establishes that blind-unforgeability is \emph{strictly stronger} than plus-one unforgeability in the quantum setting.

The implication is proven using a new \emph{smoothing} technique that coherently rescales the query histories of a quantum-query algorithm in a black-box manner. This enables us to transform any successful plus-one attacker into a blind-unforgeability attacker with only inverse-polynomial loss.  Our proof represents a conceptual shift from merely extracting a useful query history as in earlier attempts, which we show inevitably fails in general, to first coherently reshaping how query histories interfere.

\end{abstract}

\newpage
\section{Introduction}
\label{sec:intro}
\input{intro}

\section{Preliminaries}
\label{sec:prelim}
\input{prelim}


\section{BU Equals Singleton-Puncture Density Unforgeability}
\label{sec:bu=spdu}
\input{bu-equiv-spdu}

\section{$\dBU$ implies $\PO$}
\label{sec:spdu2po}
\input{dbu2po}

\section{Ruling Out Raw Singleton Extractions}
\label{sec:counterexample}
\input{counterexample}


\ifanonoymous\else
\subsubsection*{Acknowledgments.}
The authors were partially supported by the USA NSF grant no. 2054758 (CAREER). 
\fi

\newpage
\bibliographystyle{splncs04}
\bibliography{bu2po}

\appendix
\newpage
\section*{Supplementary Material}
\addcontentsline{toc}{section}{Supplementary Material}

\input{appendix}

\end{document}

%% file: intro.tex
Message authentication codes (MACs) are a fundamental primitive in
symmetric-key cryptography.  At the core, they formalize a basic
unpredictability requirement for a keyed function family \(F_k : \mathcal X \rightarrow \mathcal Y\), where $F_k(x)$ is the authentication tag of message $x$.  In the classical chosen-message attack model (see, e.g.,~\cite{BKR00}), an adversary may adaptively query $F_k(\cdot)$ and obtain a list of tags $F_k(x_1),\ldots,F_k(x_q)$. But it should nevertheless be infeasible for the adversary to output a valid pair \((x^\star,F_k(x^\star))\) for a \emph{fresh} message \(x^\star\notin\{x_1,\ldots,x_q\}\).  

The situation changes qualitatively when the adversary is granted quantum
oracle access, where a single query may have the form
\[
    \sum_x \alpha_x \ket{x,0}
    \longmapsto
    \sum_x \alpha_x \ket{x,F_k(x)} \, .
\]
Hence the classical criterion ``predict at an unqueried point'' has no literal quantum analogue, and it becomes elusive what it should mean for a classical function to remain unpredictable in the presence of quantum superposition access.

Two prominent approaches~\cite{BZ13,AMRS20} replace classical freshness in rather different ways.  Boneh and Zhandry~\cite{BZ13} proposed a \emph{counting}-based notion: after making $q$ quantum queries, an adversary should not be able to output $q+1$ distinct valid input-output pairs.  This is usually referred to as \emph{plus-one
unforgeability} (\PO), and it avoids referring to a query transcript altogether.

Alagic, Majenz, Russell, and Song~\cite{AMRS20} proposed instead
\emph{blind-unforgeability} (\BU).  A random subset
$B\subseteq\mathcal X$ is sampled by independently blinding each input with
some rate $\beta$, and the adversary receives the modified oracle
\[
    F_k^{\setminus B}(x)
    =
    \begin{cases}
        \bot, & x\in B,\\
        F_k(x), & x\notin B.
    \end{cases}
\]
The adversary wins if it outputs $(x,F_k(x))$ for a blinded point $x\in B$.
Thus BU replaces transcript-based freshness by
\emph{counterfactual withholding}: a successful forgery is fresh because
its value was deliberately unavailable from the oracle throughout the
attack.

These two approaches, \emph{counting} versus \emph{withholding}, suggest rather different interpretations of quantum
freshness.  Their relationship
has remained unresolved.  It was shown in~\cite{AMRS20} that $\PO$ does not
imply $\BU$ by exhibiting a family that is secure according to $\PO$ but admits a direct attack detected by $\BU$.  An earlier version of the same work also claimed the converse direction that $\BU$ implies $\PO$.  It was later withdrawn due to an error and was replaced by a weaker claim that $\BU$ security implies that producing $q^2$ valid pairs out of $q$ queries is difficult.  This leaves a basic question:

\begin{center}
    \emph{Do $\PO$ and $\BU$ capture genuinely incomparable aspects of quantum
    unpredictability, or is one notion strictly stronger than the other?}
\end{center}

This unsettled situation actually inspired some recent effort~\cite{JKLL26} proposing a unified definition that tunes between $\PO$ and $\BU$. 

\paragraph{Our result.}
We resolve this question by proving that blind-unforgeability indeed implies
plus-one unforgeability.  Hence the withholding protection imposed by $\BU$ strictly subsumes the counting-based protection of $\PO$.

\begin{theorem}[Informal] $\BU \Rightarrow \PO$ holds for
    any quantum-accessible keyed function family.
    \label{thm:main-informal}
\end{theorem}

In other words, denoting $\mathcal{F}_{\PO}$ and $\mathcal{F}_{\BU}$ as classes of function families satisfying each security respectively, we have that $\mathcal{F}_{\BU} \subsetneq \mathcal{F}_{\PO}$.  More specifically, suppose that a $q$-query $\PO$ adversary $A$ correctly outputs
$q+1$ distinct pairs with probability $\delta$.  We construct a
$\BU$ adversary making at most $3q+1$ queries to a blinded oracle with a suitable blinding rate that breaks $\BU$ with probability
$\Theta\left({\delta}/ {q^3}\right) - \negl(\secpar)$.  In particular, a noticeable $\PO$ advantage yields noticeable $\BU$ advantage with
only polynomial loss.

\subsection{Technical overview}
\label{sec:tech-overview}

It is helpful to first review the weaker implication in~\cite{AMRS20}, which shows
that $\BU$ rules out attacks producing \emph{quadratically} more valid
input-output pairs than the query quota.  Loosely speaking, the proof of~\cite{AMRS20} first chooses the Bernoulli
blinding rate small enough around $\beta=O(1/q^2)$ that the blinded oracle does not substantially
disturb a $q$-query adversary due to a Grover-type argument.  But then, for the adversary to hit a blinded point with
noticeable probability, it must output on the order of $1/\beta=\Omega(q^2)$
candidate pairs.  Thus the quadratic output threshold is not accidental, and it is exactly the amount needed to compensate for the small blinding rate required to preserve the computation.    

This reasoning stops at the genuine plus-one regime.  If the
adversary reports only $q+1$ points, the same blindness rate
$\beta=\Theta(1/q^2)$ hits a reported point with probability only
$O(1/q)$, and increasing $\beta$ enough to obtain a constant hit probability
causes too much disturbance to the quantum computation.  Therefore
one is led to exploit stronger combinatorial features specific to $\PO$
security, as attempted in earlier work.  Along every successful $q$-query history producing $q+1$
distinct pairs, at least one reported point was never queried.  Hence a $q$-query history
can touch at most $q$ distinct inputs.  Classically, this 
pigeonhole observation would suggest puncturing an omitted reported point and retaining the corresponding successful part of the computation. For a quantum algorithm, however, a query history means a coherent computation branch labeled by the classical sequence of query addresses along that branch.  The different histories are recombined coherently, and the ``leave-one-out'' contributions associated with different histories may cancel.

This is not merely a limitation of a particular analysis.  We construct an explicit ``all-or-nothing'' $\PO$ adversary that succeeds with probability arbitrarily close to one, yet has the following striking property: for every reported point $x$, the successful histories that never queried $x$, when coherently combined, cancel exactly. In other words, the adversary has no usable raw leave-one-out component at any reported point.  Thus even though every successful history omits some reported point, the original interference pattern of the adversary needs not contain any usable singleton leave-one-out witness.  In particular, reductions whose quantitative guarantee factors through such raw zero-touch histories, such as running the $\PO$ adversary unchanged and attempting to extract a forgery from its pre-existing leave-one-out components, can fail completely.

The lesson is that the one-unit surplus is present path by path, but may be hidden by interference when those paths are recombined.  Our reduction therefore does not try to extract the surplus from the original interference pattern.  Instead, it first \emph{reshapes} that pattern by coherently reweighting different query histories.  The resulting
\emph{Poisson smoothing} detunes the cancellations that hide the zero-touch contribution and makes a subsequent singleton extraction possible.

\paragraph{The first cancellation problem.}
Let us take a closer look at the destructive cancellation.  Consider a $\PO$ attacker $A$ with success probability $\delta$. For an intuitive explanation, fix one successful output set $X$ with $|X|=q+1$. For a query history $h$, let $M(h)\subseteq X$ denote the omitted ones that the history never queries, and necessarily $M(h)\neq\emptyset$.  We then decompose $A$'s successful state according to omitted sets. Namely, for each possible $M\subseteq X$, we coherently collect all query histories that omit exactly the subset $M$, and let $v_M$ denote the component state contributed by these histories in the final output state.  Then the part in the output state that contributes to $\PO$ success is thus 
\[  v :=\sum_{\emptyset \neq M\subseteq X} v_M \quad \text{ with } \norm{v}^2 \ge \delta \, .  \]
    
Then consider a candidate reported point $x\in X$. We can see that the histories that report $x$ but
never touch it in the queries contribute $\xi_x:=\sum_{M\ni x} v_M$. 
It is tempting to conclude that some $\xi_x$ must be large because every
successful history belongs to at least one such class.  The obstruction is
that these classes overlap.  Compare
\[
    \xi: = \sum_{x\in X} \xi_x
    =
    \sum_{\emptyset\neq M\subseteq X} |M|\,v_M , \quad \text{ v.s. } \quad v=\sum_{\emptyset\neq M\subseteq X} v_M \, .
\]
It can happen that the intact $\PO$ success from the \emph{unweighted} sum $v$ is  large, but the former \emph{weighted} sum $\xi$ cancels almost completely.  This is the
precise interference phenomenon exploited by our all-or-nothing counterexample.

\paragraph{Poisson smoothing: detuning the cancellation.}
Our main technical idea is to coherently rescale the relative weights of the different omission layers.  We transform $A$ into a randomized family of algorithms $\{A_{t,z}\}$ indexed by $t\in[0,1]$ and random phases $z = \{z_x\}_{x\in \X}$.

The transformation has a clean circuit description\footnote{Actual implementation needs care on finite-precision discretization and numerical well-conditioning which we do not elaborate in the introduction. See~\cref{sec:implementation}.}.  Before
each oracle query at $x$, we multiply the query amplitude
by a random phase $z_x^*$, and after $A$ has coherently produced its final output list $X$, we apply a compensating phase $\Pi_{x \in X} z_x$ (see~\cref{fig:poisson-smoothing}).  The random phases are chosen so that,
after averaging, a successful history omitting $r$ reported points receives
the scalar multiplier $t^{2r-1}$. Thus histories with different omission cardinalities are
reweighted differently as $t$ varies.  Note that at $t=1$ the original adversary $A$ is recovered, and should a cancellation occur in the original computation, generally it is 
detuned away from that endpoint.  To see it quantitatively, define
\[
    H(t)
    :=
    \sum_{M\neq\emptyset} t^{2|M|}v_M ,
\]
which interpolates between $H(0)=0$ {and} $H(1)=v$ (recall $v$ is the total successful state contribution of the original $\PO$ attacker).  After smoothing,  the zero-touch contribution at each reported point $x$ becomes  
\[
    \xi_x(t)
    :=
    \sum_{M\ni x} t^{2|M|-1}v_M \, .
\]
The choice of exponent is
precisely such that
\[
    \sum_{x\in X}\xi_x(t)
    =
    \frac12 H'(t) \, .
\]
Then by standard methods in analysis (e.g., Cauchy-Schwarz and energy bound for a differentiable path), we can get that 
\[
    \int_0^1
    \sum_{x\in X}\|\xi_x(t)\|^2\,dt
    \ge
    \frac{1}{4(q+1)}   \int_0^1 \|H'(t)\|^2\,dt
    \ge \frac{\|v\|^2}{4(q+1)}.
\]
Note that the left-hand side is the zero-touch mass accumulated along the interpolation: for each \(t\), we sum the squared norms of the zero-touch contributions \(\xi_x(t)\) over all reported points \(x\in X\), and then average this quantity over $t \in [0,1]$.  Hence, even if the zero-touch contributions cancels completely at a
particular value of $t$, they cannot remain small throughout the smoothing
interpolation.  Averaging over $t$ and the random phases therefore produces
a smoothed adversary $B$ with noticeable zero-touch forgery mass.  Notice that that this transformation does not measure or record a query history.  It only changes the relative phases with which histories interfere.  

Finally how to get the desired random phases?  We will sample $u_x$ uniformly on the unit circle (i.e., $u_x : = e^{2\pi i \theta_x}$ with uniform $\theta_x \in [0,1)$), and then take
\[
    z_x:=\frac{t+u_x}{1+t u_x} \, .
\]

In Fourier analysis, this coincides with the \emph{Poisson-kernel}
distribution, and $\E[z_x^j]=t^{|j|}$ is explained by the fact that the $j$th Fourier coefficient is exactly $t^{|j|}$.  For this matter, we view this step as a form of \emph{Poisson smoothing of coherent query histories}, which \emph{detunes} cancellations between different omission layers until zero-touch mass becomes visible rather than to identify or measure useful histories.

\paragraph{The second cancellation problem: making zero-touch mass
operational.}
Poisson smoothing gives an adversary $B$ with large zero-touch mass, however one cannot simply run
$B$ against a punctured oracle.  Fix a puncture $x^\star$.  Histories that never touch $x^\star$ are unaffected by the puncture, but
histories that do touch it are changed.  Their new amplitudes may again interfere destructively with the desired zero-touch sector.

This second source of cancellation is relatively easier to deal with by a \emph{root-of-unity touch-count twirl}.  During the punctured execution, we attach a random root-of-unity phase according to the number of times the history touches the puncture.  Averaging over the random phase dephases the
different touch-count compoenents $0,1,\ldots,q$.  In particular, the zero-touch histories remain intact and can no longer
destructively interfere with any positive-touch components.  Note that this procedure 
does not need to know the puncture $x^\star$.

In the end, upon obtaining $q+1$ candidate pairs, we query the punctured
oracle on the reported inputs.  The unique $\bot$ response identifies the
puncture, while the remaining pairs can be verified.  This gives a singleton-puncture completion algorithm $C$, and we can quantify its performance by the \emph{singleton-puncture completion density}
\[
    \rho(C)
    :=
    \mathbb E_k
    \sum_{x\in\mathcal X}
    \Pr\!\left[
        C^{F_k\setminus\{x\}}
        \text{ outputs }(x,F_k(x))
    \right] \, ,
\]
where $C$ is not given $x$-punctured oracle $F_k^{\setminus \{x\}}$.  The smoothing and zero-touch extraction
together will give $\rho(C)\gtrsim {\delta}/{q}$ using at most $3q+1$ oracle queries.

\paragraph{From singleton punctures to Bernoulli blinding.}
Finally, it remains to show that large singleton-puncture completion density $\rho(C)$ suffices for $C$ to break $\BU$.  This follows because singleton puncturing can be seen as the first-order form of Bernoulli blinding. The singleton-puncture density \(\rho(C)\) is the derivative at blinding rate \(\beta=0\) of the $\BU$ success probability, and polynomial degree bounds for \(O(q)\)-query algorithms imply that a large derivative forces noticeable $\BU$ success for some \(\beta\), with only an \(O(q^2)\) loss. Hence noticeable \(\rho(C)\) already breaks $\BU$, offering overall \(\BU\) advantage $\Omega(\delta/{q^3})$. In fact, this gives an equivalent characterization of $\BU$, which we formalize as \emph{singleton-puncture density unforgeability} (see~\cref{sec:bu=spdu}). 
\subsection{Discussions}
\label{sec:open}
Our quantitative reduction loses a polynomial factor in the query complexity. It would be interesting to
determining the optimal loss, or identifying which losses are inherent. In addition, beyond resolving the relationship between $\BU$ and $\PO$,
our proof highlights two broader lessons about quantum-query security.

First, successful quantum information needs not reside in a robust classical query history, but may instead be encoded in interference itself, as the all-or-nothing failure example illustrates.  Second, the proof suggests a different way to reason about such interference.  Rather than recording query histories, at the risk of destroying their coherence, one may modify them using random phases and control the resulting interpolation.  

This viewpoint is complementary to a broad line of proof techniques in the quantum-random oracle model~\cite{BDFLSZ11,Unruh14,Unruh15},  including recording and compressed-oracle methods~\cite{Zhandry19,CFHL21}, measure-and-reprogram and adaptive-reprogramming
techniques~\cite{LZ19,DFM20,GHHM21,DFMS22,CGLS24,CGLS25}. Poisson smoothing takes a different route. It does not attempt to identify or reprogram a query, but instead coherently changes the relative
weights of query histories and hence their interference pattern.  It is interesting to investigate whether other physically implementable phase distributions can be useful, where traditional query-recording or
reprogramming arguments encounter destructive interference.

\paragraph{AI disclosure.} LLM (ChatGPT Sol-5.6 high) has been used in the following occasions. 
\begin{itemize}
    \item Beyond routine typo and grammar checks, LLM was used to produce the
    figures in this paper from scans of the authors' handwritten sketches. The authors then edited and refined them.

    \item Starting from the authors' description of a function together with
    a Grover-type $\PO$ attacker that deliberately aborts on the blinding
    symbol $\bot$ and hence fails to win $\BU$,  LLM refined it into the
    complete cancellation construction of the All-or-Nothing theorem
    (\cref{thm:aon-cancellation}). The authors verified the construction and
    wrote the proof.

    \item Most substantively, the Poisson smoothing method
    (\cref{thm:singleton-puncture-smoothing}) was suggested by LLM in a
    conversation exploring how to rescale query histories so that they admit
    an efficient quantum circuit implementation. The authors verified it and developed the full proof.
\end{itemize}
The authors conducted independent verification in each case of the LLM output, and accept full responsibility for the veracity and correctness of all material in this paper.

%% file: prelim.tex
\subsection{Notations and Useful Facts}
\label{sec:basics}
We write $\poly(\secpar)$ for an unspecified polynomial in the security
parameter $\secpar$.  A function $\mu(\secpar)$ is \emph{negligible} if
$\mu(\secpar)<1/p(\secpar)$ for every polynomial $p$ and all sufficiently
large $\secpar$, and \emph{noticeable} if
$\mu(\secpar)\ge 1/\poly(\secpar)$ for all sufficiently large $\secpar$.  We write $\negl(\secpar)$ for an unspecified negligible function.  For a complex number $z\in \mathbb{C}$, $z^*$ denotes its complex conjugate. For a real-valued function
$f$, we write $f'$ for its derivative.  For a finite set $X$ and
$\beta\in[0,1]$, we write $B \xleftarrow{\beta} X$ to mean that $B\subseteq X$ is sampled according to the Bernoulli product
distribution with rate $\beta$, i.e., each $x\in X$ is included in $B$
independently with probability $\beta$.

\begin{lemma}[Markov's inequality~\cite{BE95}]
For a polynomial \(p\) of degree at most \(d\) in \([-1,1]\),
\[
\norm{p'}_\infty \le d^2 \norm{p}_\infty \, . \]
\label{lemma:markov}
\end{lemma}

A simple rescaling of it gives a useful claim for us when working with polynomials in domain \([0,1]\). 

\begin{lemma}
For a polynomial \(p\) of degree at most \(d\) in \([0,1]\),
\[
\norm{p'}_\infty \le 2d^2 \norm{p}_\infty \, . \]
\label{lemma:markov01}
\end{lemma}

\begin{proof}
    
For the interval \([0,1]\), define
\( q(x) :=p\!\left(\frac{x+1}{2}\right) \, . \)
Then \( q'(x)=\frac12p'\!\left(\frac{x+1}{2}\right), \) so Markov's inequality on \([-1,1]\) gives \(
\frac12\|p'\|_{\infty}
\le d^2\|p\|_{\infty} \, .
\)
\end{proof}

\begin{lemma}[Jensen's Inequality]
    For a convex function $f$ and a random variable $X$, 

    \begin{equation}
        \E[(f(X)] \ge f(\E [X]).
    \end{equation}
    \label{lemma:jensen's}
\end{lemma}

\begin{lemma}[Cauchy-Schwarz Inequality; discrete] \label{lem:cauch-schwarz-discrete}
    For any two vectors $\mathbf{u}$ and $\mathbf{v}$ in an inner product space, 
    \[
        \|\mathbf{u}\cdot \mathbf{v}\| \leq \|\mathbf{u}\|\cdot\|\mathbf{v}\|.
    \]
\end{lemma}

\begin{lemma}[Cauchy-Schwarz Inequality; continuous] \label{lem:cauch-schwarz-continuous}
    For functions $f, g$ integrable on $[a,b]$,
    \[
        \norm{\int_a^b f(x)\cdot g(x) dx}^2 \leq \int_a^b\norm{f(x)}^2 dx \int_a^b \norm{g(x)}^2 dx.
    \]
\end{lemma}

\subsection{Quantum Oracles and Query Algorithms}
\label{sec:quantum}

\paragraph{Query histories.}
Write the $q$-query algorithm as
\[
    A=A_q O_F A_{q-1}\cdots O_F A_0.
\]
Let $P_x$ denote the projector onto query address $x$.  Inserting
$I=\sum_x P_x$ before each oracle call gives the exact coherent
decomposition
\[
    A
    =
    \sum_{h=(x_1,\ldots,x_q)}
    A_q O_F P_{x_q}A_{q-1}\cdots O_F P_{x_1}A_0.
\]
We call $h=(x_1,\ldots,x_q)$ a \emph{query history}.  This is only an
analytic decomposition: the reduction never measures or records the query
addresses, and the corresponding history amplitudes remain coherent and may
interfere.  
\begin{definition}[History statistics]
\label{def:history-statistics}
Let $h=(x_1,\ldots,x_q)$ be a query history over the input domain
$\mathcal X$.  Its \emph{query multiplicity profile} is
\[
    \nu_h(x)
    :=
    \bigl|\{i\in[q]:x_i=x\}\bigr|,
    \qquad x\in\mathcal X.
\]
We say that $h$ \emph{touches} $x$ if $\nu_h(x)>0$.

Let $X\subseteq\mathcal X$ be a set of $q+1$ distinct reported inputs
associated with a successful history.  The \emph{omitted-output set} is
\[
    M_X(h)
    :=
    \{x\in X:\nu_h(x)=0\}.
\]
Finally, define the \emph{signed deficit vector}
\[
    \lambda_{h,X}
    :=
    \mathbf 1_X-\nu_h,
\]
where $\mathbf 1_X$ denotes the indicator vector of $X$.
Equivalently,
\[
    \lambda_{h,X}(x)
    =
    \begin{cases}
        1-\nu_h(x), & x\in X,\\
        -\nu_h(x),  & x\notin X.
    \end{cases}
\]
\end{definition}

\begin{lemma}[Deficit identities]
\label{lem:deficit-identities}
Let $h$ be a $q$-query history and let $X$ contain $q+1$ distinct
reported inputs. Then
\[
    M_X(h)\neq\varnothing,
    \qquad
    \lambda_{h,X}(x)>0
    \iff
    x\in M_X(h),
\]
and
\[
    \sum_{x\in\mathcal X}\lambda_{h,X}(x)=1,
    \qquad
    \|\lambda_{h,X}\|_1
    =
    2|M_X(h)|-1.
\]
\end{lemma}
\begin{proof}
Since $h$ contains only $q$ queries while $X$ contains $q+1$ distinct
points, at least one point of $X$ is not touched, proving
$M_X(h)\neq\varnothing$.

Moreover, $\lambda_{h,X}(x)>0$ can occur only when $x\in X$ and
$\nu_h(x)=0$, in which case $\lambda_{h,X}(x)=1$.  Hence the positive
coordinates of $\lambda_{h,X}$ are exactly those indexed by $M_X(h)$.

Since
\[
    \sum_x \nu_h(x)=q
    \qquad\text{and}\qquad
    \sum_x \mathbf 1_X(x)=|X|=q+1,
\]
we have
\[
    \sum_x \lambda_{h,X}(x)=1.
\]

The total positive mass of $\lambda_{h,X}$ is $|M_X(h)|$.  Since its
signed sum is $1$, its total negative mass in absolute value is
$|M_X(h)|-1$.  Therefore
\[
    \|\lambda_{h,X}\|_1
    =
    |M_X(h)|
    +
    \bigl(|M_X(h)|-1\bigr)
    =
    2|M_X(h)|-1.
\]
\end{proof}

The coherent history decomposition lets us isolate, for each candidate
reported point $a$, the contribution of those histories that never touch
$a$.  We then retain only successful outputs that contain $a$ and are valid
under the secret key.  Aggregating these contributions over $a$ gives what we call the
singleton zero-touch mass.

\begin{definition}[Singleton zero-touch mass]
    The filter $E_{a,0}$ extracts the sub-portion of $A$ for which $\nu_h(a) = 0$, throwing out any histories in which $a$ is a query.  The filter $\Pi_{k,a}$ projects the outputs of $A$ such that every single output-forgery pair is valid under the key $k$, and that $a$ belongs to the output set $X$. Using these, we define the \emph{singleton zero-touch mass} of $A$ as follows: 
    \begin{equation*}
        \Xi(A) := \E_k \sum_{a\in \X} \norm{\Pi_{k,a} E_{a,0} A^{F_k} \ket0}.
    \end{equation*}
\end{definition}

%% file: bu-equiv-spdu.tex
\subsection{Formulating Singleton-Puncture Density Unforgeability}
\label{sec:spdu}

Let $F=\{F_k:\X_\secpar\to\Y_\secpar\}_{k\in \K_\secpar}$ be an efficiently computable deterministic keyed family, accessed through its
standard quantum unitary oracle.  For \(a\in\X_\secpar\), let
\[
 F_k^{\setminus a}(x)=
 \begin{cases}
  \bot,&x=a \, ,\\
  F_k(x),&x\ne a \, .
 \end{cases}
\]

\begin{definition}[Singleton-puncture density unforgeability]
For a quantum algorithm \(A\), define
\[
 p_{k,a}(A)
 :=
 \Prb\!\left[A^{F_k^{\setminus a}}
      \text{ outputs }(a,F_k(a))\right]
\]
and
\begin{equation}
 \rho_F(A)
 :=
 \E_k\sum_{a\in\X_\secpar}p_{k,a}(A) \, .
 \label{eq:density}
\end{equation}
The family \(F\) is \emph{density unforgeable under singleton puncture}, or
\(\dBU\)-secure, if \(\rho_F(A)\) is negligible for every QPT \(A\).
\end{definition}

We will prove that \(\dBU\) gives an equivalent characterization of \(\BU\), and the two directions are shown in \cref{sec:dbu2bu} and \cref{sec:bu2dbu} respectively. 
\begin{theorem}For every efficiently computable keyed family \(F\),
\[
 F\text{ is }\BU\text{-secure}
 \quad\Longleftrightarrow\quad
 F\text{ is }\dBU\text{-secure}.
\]
\label{thm:bueqdbu}
\end{theorem}

In the remainder of this subsection, we make a few remarks to explain and justify the choices in this definition, which will be useful in the equivalence proof as well. We first describe some algebraic structure of a \(q\)-query quantum algorithm in a \(\BU\) game and its singleton-puncture density.

\begin{proposition}
Let \(A\) be a \(q\)-query algorithm in the \(\BU\) game. Define 
\[ G_{A,k} (\beta) : = \Pr[ A^{F_k^{{\setminus B}}} \text{ wins } \BU: B\gets \beta(X_\secpar)] \, .\]
Then \( G_A(\beta) : = \E_k\!\left[G_{A,k}(\beta)\right]\) is a polynomial in \(\beta\) of degree at most \(d = 2q+1 \), and moreover
\[ G_A(0)=0, \qquad
 G_A'(0)=\rho_F(A) \, .\]
\label{prop:polychar}
\end{proposition}

The proof is reminiscent of the analysis of the quantum polynomial method~\cite{BBCMdW01} and we defer it to~\cref{sec:polychar}. 

\begin{remark}
\cref{prop:polychar} explains the symbol \(\partial\), which indicates that \(\rho_F(A)\) is a derivative rather
than a probability. In other words, we can view \(\rho_F(A)\) as the first-order ``sensitivity'' of \(\BU\) success to the blind density.
\end{remark}

\begin{remark}[Why not normalize by a uniform puncture?]
It is tempting to think of \(\rho_F(A)\) as an \emph{unnormalized} success of winning \(\BU\) conditioned on exactly one point being blind. Namely, let \(N_\secpar=|\X_\secpar|\). If the challenger samples a hidden uniform
\(a\leftarrow\X_\secpar\), gives \(F_k^{\setminus a}\), and declares success
when \(A\) outputs \((a,F_k(a))\). We may call this singleton-blind unforgeability (SBU) game, and denote the success probability as \(\operatorname{Adv}^{\mathrm{SBU}_F}(A)\), which is
\[
 \operatorname{Adv}^{\mathrm{SBU}_F}(A)
 =\frac{\rho_F(A)}{N_\secpar}\, .
\]
The normalization factor \(\frac{1}{N_\secpar}\) is harmless when \(N_\secpar\) is polynomial, but on an
exponential domain a noticeable \(\rho_F(A)\) becomes an
exponentially small success probability.  Indeed, if \(A\) makes \(q\)
queries, then ~\cref{prop:polychar} tells us that \(G_A(\beta)\) has degree at most \(2q+1\) in \(\beta\) and \(\rho_F(A) = G_A^{'}(0)\), and Markov's inequality~(\cref{lemma:markov01})
gives
\[
 0\le \rho_F(A)\le 2(2q+1)^2 \, .
\]
Consequently
\[
 \operatorname{Adv}^{\mathrm{SBU}_F}(A)
 \le \frac{2(2q+1)^2}{N_\secpar} \, ,
\]
regardless of the function family. The unnormalized density in
\eqref{eq:density}, not the uniform-singleton probability, is therefore the
correct intermediate object that faithfully captures the information needed to show \(\BU\Rightarrow\mathsf{PO}\).
\end{remark}

\subsection{\(\dBU\) implies \(\BU\)}
\label{sec:dbu2bu}

\begin{lemma}[\(\dBU \implies \BU\)]
Let \(A\) make \(q\) queries in the BU game with Bernoulli blinding rate 
\(\beta \in (0,1)\).  There is an efficient singleton-puncture algorithm
\(C_{A,\beta}\), making at most \(2q\) queries, such that
\begin{equation}
 \rho_F(C_{A,\beta})=\frac{G_A(\beta)}{\beta}.
 \label{eq:palm}
\end{equation}
\end{lemma}

\begin{proof}
For fixed \(k,B,a\), let
\[
 p_A(k,B,a)
 :=\Prb\!\left[A^{F_k^{\setminus B}}
               \text{ outputs }(a,F_k(a))\right].
\]
Given the singleton-punctured oracle \(F_k^{\setminus a}\), the algorithm
\(C_{A,\beta}\) samples a blinding set
\(D \) by including every \(x\in \X_\secpar\setminus\{a\}\) in \(D\) with probability \(\beta\) \emph{independently}, and presents \(A\) with the further-blinded
oracle \(F_k^{\setminus(D\cup\{a\})}\). This can be implemented coherently and efficiently: query \(F_k^{\setminus a}(x)\) with fresh ancilla as the output register, 
compute a mask bit \(b_x : = x \overset{?}{\in} {D} \cup \{a\}\) in another ancilla qubit, then write either \(\bot\) or the function value into \(A\)'s answer
register controlled by all the ancilla, and then uncomputing the ancilla.  This uses two calls to the base oracle \(F_k^{\setminus a}(x)\).
Since \(\) is a \(q\)-query algorithm, $D$ can be sampled by a \(2q\)-wise independent hash function. In the end, \(C_{A,\beta}\) returns \(A\)'s output.  We can see that 
\begin{align*}
 \rho_F(C_{A,\beta})
 &=\E_k\sum_a\sum_{B\ni a}
    \beta^{|B|-1}(1-\beta)^{N_\secpar-|B|}p_A(k,B,a)\\
 &=\frac1\beta\E_k\sum_B
    \beta^{|B|}(1-\beta)^{N_\secpar-|B|}
    \sum_{a\in B}p_A(k,B,a)\\
 &=\frac{G_A(\beta)}\beta.
\end{align*}
Therefore, whenever \(A\) breaks \(\BU\) with success \(G_A(\beta) \ge \frac{1}{\poly(\secpar)}\), \(\rho_F(C_{A,\beta}) \ge G_A(\beta) /\beta \ge \frac{1}{\poly(\secpar)}\) as a result for \(\beta \in (0,1)\). 
\end{proof}

\subsection{\(\BU\) implies \(\dBU\)}
\label{sec:bu2dbu}
\begin{lemma}[\(\BU \implies \dBU\)]
\label{lemma:bu2dbu}
Let \(A\) be a \(q\)-query algorithm in the \(\dBU\) game, and let
\(\rho_F(A)\) denote its singleton-puncture density.  Let
\[
    G_A(\beta)
    :=
    \Pr\!\left[A \text{ wins the } \BU
        \text{ game with blinding rate } \beta\right].
\]
Then
\[
    \sup_{\beta\in[0,1]} G_A(\beta)
    \ge
    \frac{\rho_F(A)}{2(2q+1)^2}.
\]
Moreover, for the explicit choice
\[
    \beta_0
    :=
    \frac{\rho_F(A)}{4(2q+1)^4},
\]
we have
\[
    G_A(\beta_0)
    \ge
    \frac{\rho_F(A)^2}{8(2q+1)^4}.
\]
Thus noticeable singleton-puncture density implies noticeable
\(\BU\) advantage; the first bound incurs only an \(O(q^2)\) loss,
while the second gives an explicit, though weaker, blinding rate.
\end{lemma}

\begin{proof} Apply Taylor's formula to \(G_A(\beta)\), we have that  
\[G_A(\beta) = G_A(0) + G_A^{'}(0)\cdot \beta + \frac{1}{2}G_A^{''}(\xi)\cdot \beta^2 \]
for some \(\xi \in (0, \beta)\). 

Notice that \(G_A(0) = 0\), and \(G_A^{'}(0) = \rho_F(A)\). Since \(G_A(\beta)\) is a polynomial of degree at most \( 2q+1\), Markov's inequality~\cref{lemma:markov01} gives that 
\[ \norm{G_A^{'}(\beta)}_\infty \le 2(2q+1)^2 \norm{G_A(\beta)}_\infty \le  2(2q+1)^2 \, .\]

Hence
\[
    \sup_{\beta\in[0,1]}G_A(\beta)
    =
    \|G_A\|_\infty
    \ge
    \frac{\rho_F(A)}{2(2q+1)^2}.
\]

Apply Markov's inequality again to \(G_A^{'}(\beta)\), which is a polynomial of degree at most \(2q\), we get that
\[ \norm{G_A^{''}(\beta)}_\infty \le 2(2q)^2 \norm{G_A^{'}(\beta)}_\infty \le 4(2q+1)^4 \, .\]
Then we can see that 
\[ G_A(\beta) \ge \rho_F(A) \cdot \beta  - \frac{1}{2} \norm{G_A^{''}(\beta)}_\infty \cdot \beta^2  \ge \rho_F(A) \cdot \beta  - 2(2q+1)^4\cdot \beta^2  \, .\]
Therefore, at \( \beta = \frac{\rho_F(A)}{4(2q+1)^4}\), 
\[ G_A(\beta) \ge \frac{\rho^2_F(A)}{8(2q+1)^4} \, .\]
Note that the chosen \(\beta\) lies in \([0,1]\) because \(0 \le \rho_F(A) \le 2(2q+1)^2 \). As a result, whenever \(\rho_F(A)\) is noticeable, \(A\) gives rise to a valid \(\BU\) adversary with noticeable winning probability. 
\end{proof}

%% file: dbu2po.tex
We show that singleton-punctured density unforgeability ($\dBU$) implies $\PO$ in this section, which in turn will establish our main theorem that $\BU$ implies $\PO$. Given a $\PO$ adversary $A$, we will show a reduction that converts $A$ to an attacker against $\dBU$. 

\subsection{PO Adversary Smoothing}
\label{sec:smoothpo}
\input{smoothpo}

\subsection{Singleton Density Extraction}
\label{sec:singleton-ext}
\input{singleton-ext}

\subsection{Finite-Precision Efficient Implementation}
\label{sec:implementation}
\input{implement}

\subsection{Putting it Together}
\label{sec:spdu2po-proof}
\input{spdu2po}

%% file: smoothpo.tex
First, choose a value $t$ uniformly from $[0,1]$ and select values of $z_x$ for every $x\in \X$ such that $\E_z [z_x^j] = t^{|j|}$ for all integers $j$. This gives us a set of randomized phases which on expectation balance out around our selected $t$. We will show that such a set of phases is constructible in~\cref{sec:implementation}. 

\paragraph{Smoothing adversary $A_{t,z}$.} Each time our adversary $A$ would make a query $|x\rangle$, we modify it by applying the appropriate phases, giving us $z_x^*|x\rangle$ instead. As this is based on the query instead of the response, it does not require any oracle input or any external knowledge to calculate. Then, when $A$ outputs its final value, we modify it with the phase as follows:

\begin{equation*}
    |(x_1,y_1), \ldots, (x_{q+1}, y_{q+1}), \ldots \rangle \to \prod_{i=1}^{q+1} z_{x_i} |(x_1,y_1), \ldots, (x_{q+1}, y_{q+1}), 
    \ldots \rangle \, .
\end{equation*}

The result of this modification is that every message which appears in the queries has a nonpositive phase equal to $\lambda[h,X](x)$.  We call this modified version of the adversary ${A}_{t,z}$, and display a circuit implementation of it in Figure \ref{fig:poisson-smoothing}. 

\input{smooth-ckt}

\begin{theorem}
    \label{thm:singleton-puncture-smoothing} 
    Let $A$ be a $q$-query $\PO$ adversary with success rate $\delta$.  Then for $t$ uniform on $[0,1]$ and $z = (z_x)_{x\in \X}$ such that $\E_z[z_x^j] = t^{|j|} \,\,\,\forall j\in\mathbb{Z}$, 
    \begin{equation*}
        \E_{t,z} \Xi( A_{t,z}) \ge \frac{\delta}{4(q+1)}
    \end{equation*}
\end{theorem}

Let us first develop several useful claims towards proving the theorem.  We define $z^{\lambda_{h,X}} := \prod_xz_x^{\lambda_{h,X}(x)}$. If we take the expectation over all possible choices of $z$, we get \begin{equation*}  \E_z\!\left[z^{\lambda_{h,X}}\right]
    =
    t^{\|\lambda_{h,X}\|_1}
    =
    t^{2|M_X(h)|-1}.
\end{equation*}

For a fixed key $k$ and output set $X$, if we choose an omitted set $M \subseteq X$, the superposition of all histories that are correct under $k$ whose outputs are equal to $X$ such that $M(h,X)=M$ is called $v_{M}$. 

 \begin{equation*}
     v_{M} := \sum_{h:M_X(h)=M} \alpha_h\ket{h}.
 \end{equation*}

We will consider this superposition over all such exclusion sets, which we do by introducing a particular formula that allows us to achieve multiple goals at once---the other purpose of this will be shown later in Lemma \ref{lemma:quarters}.

\begin{align*} 
    H(t) :=& \sum_{M\ne\emptyset} t^{2|M|} v_{M}.\\
    H(0) = 0,& \mspace{20 mu} H(1) = v.
\end{align*}

This new $v$ is the superposition over all histories $h$ that contribute to an output $X$ which is correct for key $k$. However, we can also add together only a specific set of these $v_M$ together in order to get a value which is useful for analyzing the single-touch case. 

\begin{equation*}
    U_a(t) := \sum_{M\owns a}  t^{2|M|-1} v_M
\end{equation*}

This value is the $t$-adjusted superposition over all successful histories for whom the message $a$ is output without having been queried. In other words, it's every possible way that $A$ could have theoretically forged $a$ as the plus-one message, with priority given to those histories where the number of forgeries is low. After all, since $t \in [0,1]$, larger sets $M$ are made more unlikely by this rescaling.   

As $t$ is an abstract value that doesn't actually touch the construction of ${A}_{t,z}$, we also create a second value which is a close approximation of the first, but which clings much closer to our actual definition for the modified adversary.

\begin{equation*}
    u_a (z) := \sum_{h:\nu_h(a) = 0} \prod_{x\in\X} z_x^{\lambda_{h,X}(x)} \alpha_h\ket h. 
\end{equation*}

Our choice of $z$ makes $U_a (t) = \E_z u_a (z)$, and therefore Lemma \ref{lemma:jensen's}  gives that

\[ \norm{U_a (t)}^2  = \norm{\E_z u_a (z)}^2 \le \E_z \norm{u_a (z)}^2 \, . \]

With these two sums that focus on a specific $a$, we'd like to see how they compare to the unfocused sum. This is the point where our previous definition of $H$ becomes useful, as it lets us draw a direct link between the derivative $H'$ and the total over all $a$ of the $U_a$ values. 

\begin{lemma}\label{lemma:quarters}
    For every fixed $k, X$, 
\begin{equation*}
        \sum_{a\in X} ||U_a (t)||^2 \ge \frac{1}{4(q+1)} \left\lVert H' (t) \right\rVert^2 \, . 
    \end{equation*}
\end{lemma}

\begin{proof} First we unpack the definition of $U_a$ and obtain that
 \begin{align*}
        \sum_{a\in X} U_a (t) &= \sum_{a\in X} \sum_{M\owns a} t^{2|M|-1} v_{k, \vec x, M} \\ 
        &= \sum_{M\neq \emptyset} |M| t^{2|M|-1 }v_M \\
        &= \frac{1}{2} H'(t) \, . 
    \end{align*}
    
   Then Cauchy-Schwarz (~\cref{lem:cauch-schwarz-discrete}) over the $q+1$ messages in $X$ gives that, 

    \begin{equation*}
        \sum_{a\in X} \lVert U_a (t)\rVert^2 \ge \frac{1}{q+1} \left\lVert \sum_{a\in X} U_a(t) \right\rVert^2 \ge \frac{1}{4(q+1)} \left\lVert H' (t) \right\rVert^2\, .  
    \end{equation*}
\end{proof}

We are now ready to prove~\cref{thm:singleton-puncture-smoothing}. Here we include a summary of symbols that have been defined for the proof for easy lookup. 

\begin{table}[ht]
\centering
\caption{Notation used in the Poisson-smoothing argument.}
\label{tab:smoothing-notation}
\begin{tabular}{c|p{9cm}}
    \hline
    Symbol & Definition \\
    \hline
    $t\in [0,1)$
    & A uniformly random scalar controlling the smoothing process. \\
    \hline
    $z=\{z_x\}_{x\in\X}$
    & A collection of random phases for smoothing. \\
    \hline
    ${A}_{t,z}$
    & The smoothed adversary   from $A$ using parameters $(t,z)$. \\
    \hline
    $v_M$
    & The state contribution of histories that omit exactly the set $M$. \\
    \hline
    $v$
    & The total successful state contribution of the original PO adversary $A$. \\
        \hline
    $\Xi(A)$
    & The overall zero-touch success mass of an adversary $A$. \\
    \hline
    $H(t)$
    & An interpolation satisfying $H(0)=0$ and $H(1)=v$. \\
    \hline
    $U_a(t)$
    & The smoothed zero-touch contribution associated with point $a$ at parameter $t$. \\
    \hline
    $u_a(z)$
    & The corresponding zero-touch contribution expressed using the phase collection $z$. \\
    \hline
\end{tabular}
\end{table}

\begin{proof}[Proof of~\cref{thm:singleton-puncture-smoothing}]
   Fix $k, X$ and observe that \begin{align*}
        v &= H(1) - H(0) = \int_0^1 H' (t) dt. 
    \end{align*}

    Taking the squared norm on both sides, and then use Lemma \ref{lem:cauch-schwarz-continuous} to move the integral out of the squared norm.  
    \begin{equation*}
        \int_0^1 \lVert H' (t) \rVert^2 dt \ge \left\lVert \int_0^1 H'(t) dt \right \rVert^2 = \lVert v\rVert^2. 
    \end{equation*}
   We can then go through the following chain of logic: 
    \begin{align*}
        \int_0^1  \sum_{a\in X}\E_z\lVert u_a (z) \rVert^2 dt &\ge \int_{0}^1 \sum_{a\in X} \left\lVert U_a(t) \right\rVert^2 dt\\
        &\ge \int_0^1 \frac{1}{4(q+1)} \lVert H' (t) \rVert^2 dt \\
        &= \frac{1}{4(q+1)} \int_0^1\lVert H'(t) \rVert ^2 dt \\
        &\ge \frac{1}{4(q+1)} \lVert v\rVert^2. 
    \end{align*}
    Recall that we originally fixed a certain $\vec x$, which must be disjoint with all other possible $\vec x$. At this point, we sum over all $\vec x$, which is possible due to the disjoint, and then we take the expected value for all $k$ to get a measurement of the success of $A_z$. We do this on both sides, and then show the conversion into our claim. 
    First, let's show that the left hand side of the inequality matches our goal. 
    \begin{align*}
       & \E_k \sum_{\vec x} \int_0^1  \sum_{a\in \vec x}\E_z\lVert u_a (z) \rVert^2 dt \\
        =&   \int_0^1\E_z\E_k \sum_{\vec x} \sum_{a\in \vec x}\left\lVert u_a (z) \right\rVert^2 dt\\
        =& \E_{t,z} \E_k \sum_{\vec x} \sum_{a\in \vec x}\left\lVert u_a (z) \right\rVert^2.
        \end{align*}

        In the first step, we rearrange the terms, passing the expectation over $z$ outside of the sums, and pulling the integral with respect to $t$ all the way out. In the next step, we take advantage of the fact that $t$ is being chosen uniformly over $[0,1]$. As each value of $t$ is equally likely to be selected, the integral from $0$ to $1$ over $t$ is exactly the same as the expectation over $t$, which lets us convert it into $\E_{t,z}$. Next, we swap the two sums, exchanging from each $X$ and each $a \in X$ to starting from each $a$ and then selecting over each $X$ which contains $a$. This is allowed because the two sum orders lead to the same total so long as it's the case that the total number of possible $X$ is finite, which it is since the length is fixed and the size of $\X$ is finite. Continuing from above, we have that
        
        \begin{align*}
        & \E_{t,z} \E_k \sum_{a\in \X} \sum_{X \owns a}\left\lVert u_a (z) \right\rVert^2\\
        =& \E_{t,z} \E_k \sum_{a \in \X} \norm{ \sum_{\vec x \owns a} u_a(z)}^2\\
        =&  \E_{t,z} \E_k \sum_{a \in \X} \norm{ \sum_{X \owns a} \sum_{h:\nu_h(a) = 0} \prod_{x\in\X} z_x^{\lambda_{h,X}(x)} \alpha_h\ket h}^2.
        \end{align*} 

        After swapping the sums, we note that since every single $X$ is orthogonal to each other one, the sum over the squared norm of the terms is the same as the squared norm of the sum, allowing us to slide the sum into the squared norm. Then we unpack the definition of $u_a$, unfolding the equation to give a complete description of what we're doing.

        Finally, we observe that the inside of our norm is an exact match for the terms we used to define $\Xi(A)$, in that on the inside is the sum of phases which correspond to the adversary, modified using our phase-smoothing trick. Then, there's something that focuses only on the histories where $a$ is never queried. Then, there's a focus that causes only the histories that output $a$ to be kept. As it turns out, the two equations are one-to-one identical---all nonzero terms are shared across the two sums, without any repeats. 
        
        \begin{align*}
        =& \E_{t,z} \E_k \sum_{a\in \X} \lVert \Pi_{k, a} E_{a,0} A_{t,z}^{F_k} \ket{0} \rVert^2\\
        =& \E_{t,z} \Xi(\tilde A_{t,z}). 
    \end{align*}

    Since we have now shown the left hand side of the inequality corresponds with our goal, we turn to the right hand side to see if things keep working properly.  

    \begin{align*}
        \E_k \sum_{X}\frac{1}{4(q+1)} \lVert v\rVert^2  &= \frac{1}{4(q+1)}\E_k \sum_{X} \lVert v\rVert^2\\ 
        &= \frac{\delta}{4(q+1)}.
    \end{align*}

    \begin{remark}
        We are able to replace the second part of the term with $\delta$ because the expectation over all keys of the sum of the probability for all successful outputs is equal to the chance of success of $A$.
    \end{remark}
    
    Since both the left hand side and right hand side reduce to our desired output, we have proven our claim.  
\end{proof}

%% file: smooth-ckt.tex
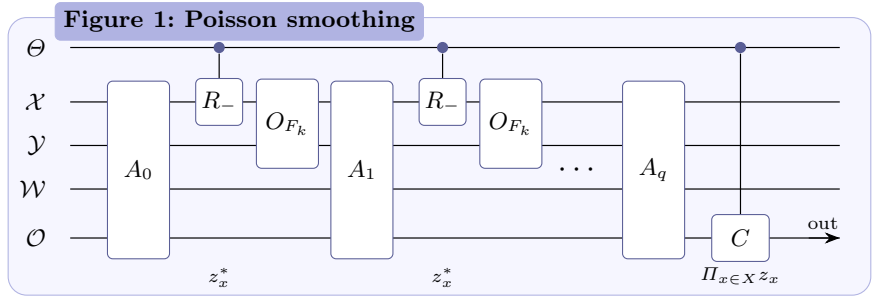
\begin{figure}[ht]
\centering
\begin{tikzpicture}[
  x=0.82cm,y=0.72cm,
  wire/.style={line width=0.45pt},
  gate/.style={
      draw=gateborder,
      fill=white,
      rounded corners=2pt,
      minimum width=0.72cm,
      minimum height=0.62cm,
      inner sep=1.5pt,
      font=\small
  },
  phase/.style={
      gate,
      minimum width=0.62cm
  },
  oracle/.style={
      gate,
      minimum width=0.82cm,
      minimum height=1.18cm
  },
  biggate/.style={
      gate,
      minimum width=0.82cm,
      minimum height=2.35cm
  },
  finalgate/.style={
      gate,
      minimum width=0.76cm
  },
  ctrl/.style={
      circle,
      fill=gateborder,
      inner sep=0pt,
      minimum size=4.2pt
  },
  lab/.style={
      font=\small,
      anchor=east
  }
]

\def\yth{0}
\def\yi{-1.0}
\def\yy{-1.8}
\def\yw{-2.6}
\def\yo{-3.5}

\node[
    draw=figborder,
    fill=figbg,
    rounded corners=7pt,
    fit={(-0.75,0.55) (13.15,-4.55)},
    inner sep=0pt
] (panel) {};

\node[
    draw=figborder,
    fill=figborder,
    rounded corners=2pt,
    anchor=west,
    font=\bfseries\footnotesize
] at (-0,0.48) {Figure~\ref{fig:poisson-smoothing}: Poisson smoothing};

\node[lab] at (0,\yth) {$\Theta$};
\node[lab] at (0,\yi) {$\mathcal X$};
\node[lab] at (0,\yy) {$\mathcal Y$};
\node[lab] at (0,\yw) {$\mathcal W$};
\node[lab] at (0,\yo) {$\mathcal O$};

\foreach \y in {\yth,\yi,\yy,\yw,\yo}
  \draw[wire] (0.25,\y) -- (12.65,\y);

\node[biggate] (A0)
    at (1.35,{(\yi+\yo)/2})
    {$A_0$};

\node[phase] (R0)
    at (2.65,\yi)
    {$R_{-}$};

\node[ctrl] (c0)
    at (2.65,\yth) {};

\draw[wire] (c0) -- (R0);

\node[oracle] (O0)
    at (3.75,{(\yi+\yy)/2})
    {$O_{F_k}$};

\node[biggate] (A1)
    at (4.95,{(\yi+\yo)/2})
    {$A_1$};

\node[phase] (R1)
    at (6.25,\yi)
    {$R_{-}$};

\node[ctrl] (c1)
    at (6.25,\yth) {};

\draw[wire] (c1) -- (R1);

\node[oracle] (O1)
    at (7.35,{(\yi+\yy)/2})
    {$O_{F_k}$};

\node[font=\large]
    at (8.45,{(\yi+\yo)/2})
    {$\cdots$};

\node[biggate] (Aq)
    at (9.65,{(\yi+\yo)/2})
    {$A_q$};

\node[finalgate] (C)
    at (11.05,\yo)
    {$C$};

\node[ctrl] (cf)
    at (11.05,\yth) {};

\draw[wire] (cf) -- (C);

\draw[-{Stealth[length=2mm]},wire]
    (12.15,\yo) -- (12.65,\yo);

\node[font=\scriptsize,anchor=south]
    at (12.38,\yo) {out};

\node[font=\scriptsize,anchor=north]
    at (2.65,-3.88)
    {$z_x^*$};

\node[font=\scriptsize,anchor=north]
    at (6.25,-3.88)
    {$z_x^*$};

\node[font=\scriptsize,anchor=north]
    at (11.05,-3.88)
    {$\Pi_{x\in X }z_x$};

\end{tikzpicture}

\caption{
Circuit implementation of Poisson phase smoothing.
The random phase data $\vec{z}=(z_x)$ is stored in the register $\Theta$.
Before each oracle call, $\Theta$ controls the address phase
\(R_{-}\lvert x\rangle = z_x^*\lvert x\rangle.\)
After $A_q$ has coherently produced the output address set $X$,
the final controlled compensation contributes
$\Pi_{x\in X} z_x$.
Thus a history with query-count vector $\nu$ acquires net phase
\(\langle \vec{z},\mathbf 1_X-\nu\rangle\).
}
\label{fig:poisson-smoothing}
\end{figure}

%% file: singleton-ext.tex
We want to use the theorem we just proved to convert $\tilde{A}_{t,z}$, a plus-one adversary smoothed using our technique, into some sort of singleton-puncture adversary. We do this using the following claim.

\begin{lemma}
    Let $B$ be a $q$-query adversary against the $\PO$ game, and let $F_k^{\setminus a}$ be an oracle with an unknown singleton puncture $a$. There exists an algorithm $C$ that uses at most $3q+1$ queries such that 
    \begin{equation}
       \rho_F(C) = \E_k \sum_{a\in \X} \Pr [C^{F_k^{\backslash a}}\, outputs\, (a, F_k(a))] \ge \Xi(B) . 
    \end{equation}
    \label{lemma:singlton-ext}
\end{lemma}

\begin{proof}
    Let $\omega = e^{2\pi i/(q+1)}$ be the $q+1$-th root of unity. $C$ selects a random value $r$ uniformly from $[0, q]$. Then, $C$ simulates the oracle for $B$ by performing each query that $B$ makes twice. One of the two responses is passed into $B$, while the other is stored inside of $C$. If the response turns out to be the punctured value, then multiply by the phase $\omega^r$. As a result, a history that queries the punctured point $s$ times is given phase $\omega^{rs}$. 

    \input{extraction-ckt}

    We are now able to filter for just the histories which never query the punctured point, since the phase of those histories is equal to $1$ while all other histories have some non-$1$ phase as $s < q+1$. 

    This gives us a set of histories which behave the same whether or not the oracle is punctured, since they never query the oracle to begin with. Now, $B$ will return a set of $q+1$ message-forgery pairs, and $C$ queries each one of the messages against the oracle. If none of them were punctured, then $C$ rejects due to having failed. However, if one of the messages returns $\bot$, then $C$ will return that message and its associated forgery. 

    This leads to a total of $3q+1$ queries made by $C$, since it first runs all $q$ queries from $B$ twice giving $2q$, then it tests all $q+1$ outputs to see if any of them were the punctured value. 

    As shown in Figure \ref{fig:singleton-extraction}, it is possible to implement this entire protocol exactly as described on a quantum circuit, showing that $C_B$ exists. 
    Furthermore, since every single time that $B$ succeeds, $C$ also succeeds, we have proof that the success chance is at least as large.

\end{proof}

%% file: extraction-ckt.tex
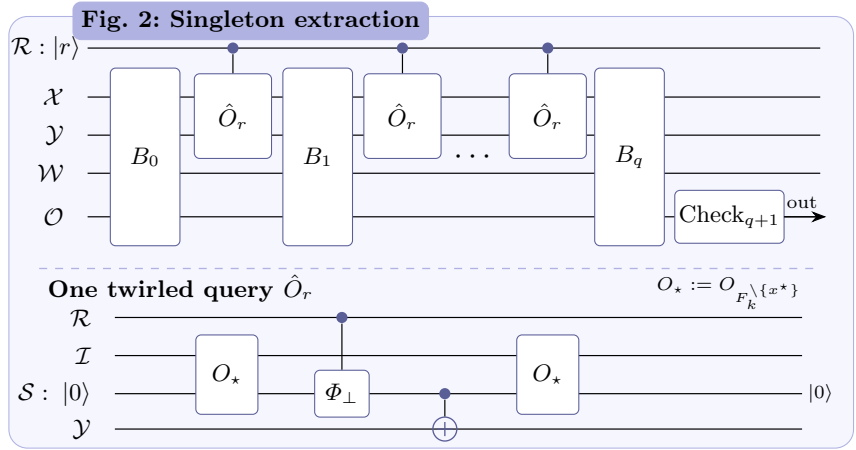
\begin{figure}[ht]
\centering
\begin{tikzpicture}[
  x=0.80cm,y=0.72cm,
  wire/.style={line width=0.45pt},
  gate/.style={
      draw=gateborder,
      fill=white,
      rounded corners=2pt,
      minimum width=0.82cm,
      minimum height=0.62cm,
      inner sep=1.5pt,
      font=\small
  },
  twirl/.style={
      gate,
      minimum width=1.02cm,
      minimum height=1.12cm
  },
  biggate/.style={
      gate,
      minimum width=0.92cm,
      minimum height=2.35cm
  },
  checkgate/.style={
      gate,
      minimum width=1.38cm,
      minimum height=0.70cm
  },
  oracle/.style={
      gate,
      minimum width=0.82cm,
      minimum height=1.05cm
  },
  phase/.style={
      gate,
      minimum width=0.72cm
  },
  ctrl/.style={
      circle,
      fill=gateborder,
      inner sep=0pt,
      minimum size=4.2pt
  },
  target/.style={
      circle,
      draw=gateborder,
      line width=0.5pt,
      inner sep=0pt,
      minimum size=9pt
  },
  lab/.style={
      font=\small,
      anchor=east
  }
]


\def\yr{0}
\def\yi{-0.9}
\def\yy{-1.6}
\def\yw{-2.3}
\def\yo{-3.1}

\node[
    draw=figborder,
    fill=figbg,
    rounded corners=7pt,
    fit={(-1.05,0.58) (12.90,-7.35)},
    inner sep=0pt
] (panel) {};

\node[
    draw=figborder,
    fill=figborder,
    rounded corners=3pt,
    anchor=west,
    font=\bfseries\footnotesize
] at (-0,0.50)
{Fig.~\ref{fig:singleton-extraction}: Singleton extraction};

\node[lab] at (0.28,\yr) {$\mathcal R:\lvert r\rangle$};
\node[lab] at (0,\yi) {$\mathcal X$};
\node[lab] at (0,\yy) {$\mathcal Y$};
\node[lab] at (0,\yw) {$\mathcal W$};
\node[lab] at (0,\yo) {$\mathcal O$};

\foreach \y in {\yr,\yi,\yy,\yw,\yo}
  \draw[wire] (0.25,\y) -- (12.35,\y);

\node[biggate]
    at (1.20,{(\yi+\yo)/2})
    {$B_0$};

\node[twirl] (T1)
    at (2.65,{(\yi+\yy)/2})
    {$\hat O_r$};

\node[ctrl] (c1)
    at (2.65,\yr) {};

\draw[wire] (c1) -- (T1.north);

\node[biggate]
    at (4.05,{(\yi+\yo)/2})
    {$B_1$};

\node[twirl] (T2)
    at (5.45,{(\yi+\yy)/2})
    {$\hat O_r$};

\node[ctrl] (c2)
    at (5.45,\yr) {};

\draw[wire] (c2) -- (T2.north);

\node[font=\large]
    at (6.65,{(\yi+\yo)/2})
    {$\cdots$};

\node[twirl] (Tq)
    at (7.85,{(\yi+\yy)/2})
    {$\hat O_r$};

\node[ctrl] (cq)
    at (7.85,\yr) {};

\draw[wire] (cq) -- (Tq.north);

\node[biggate]
    at (9.20,{(\yi+\yo)/2})
    {$B_q$};

\node[checkgate] (Chk)
    at (10.85,\yo)
    {$\mathrm{Check}_{q+1}$};

\draw[-{Stealth[length=2mm]},wire]
    (11.75,\yo) -- (12.45,\yo);

\node[font=\scriptsize,anchor=south]
    at (12.05,\yo) {out};


\draw[
    dashed,
    draw=figborder,
    line width=0.5pt
]
    (-0.55,-4.05) -- (12.35,-4.05);


\node[
    font=\bfseries\footnotesize,
    anchor=west
]
    at (-0.55,-4.40)
    {One twirled query $\hat O_r$};

\node[
    font=\scriptsize,
    anchor=east
]
    at (12.20,-4.50)
    {$O_\star:=O_{F_k^{\setminus\{x^\star\}}}$};

\def\ir{-4.95}
\def\ii{-5.65}
\def\is{-6.35}
\def\iy{-7.00}

\node[lab] at (0.45,\ir) {$\mathcal R$};
\node[lab] at (0.45,\ii) {$\mathcal I$};
\node[lab] at (0.45,\is) {$\mathcal S:\;\lvert 0\rangle$};
\node[lab] at (0.45,\iy) {$\mathcal Y$};

\foreach \y in {\ir,\ii,\is,\iy}
  \draw[wire] (0.70,\y) -- (12.05,\y);

\node[oracle] (Q1)
    at (2.55,{(\ii+\is)/2})
    {$O_\star$};

\node[phase] (Phi)
    at (4.45,\is)
    {$\Phi_{\bot}$};

\node[ctrl] (cr)
    at (4.45,\ir) {};

\draw[wire] (cr) -- (Phi.north);

\node[ctrl] (cs)
    at (6.15,\is) {};

\node[target] (ty)
    at (6.15,\iy) {};

\draw[wire] (cs) -- (ty);

\draw[gateborder,line width=0.5pt]
    ($(ty.center)+(-0.13,0)$)
    --
    ($(ty.center)+(0.13,0)$);

\draw[gateborder,line width=0.5pt]
    ($(ty.center)+(0,-0.13)$)
    --
    ($(ty.center)+(0,0.13)$);

\node[oracle] (Q2)
    at (7.85,{(\ii+\is)/2})
    {$O_\star$};

\node[
    font=\scriptsize,
    anchor=west
]
    at (12.00,\is)
    {$\lvert 0\rangle$};

\end{tikzpicture}

\caption{
Singleton zero-touch extraction by root-of-unity touch-count twirling.
The unknown puncture is $x^\star$.
Choose $r\in\mathbb Z_{q+1}$ uniformly, let
$\omega=e^{2\pi i/(q+1)}$, and reuse the same $r$ in all $q$
simulated queries of $B$.
Each $\hat O_r$ uses two calls to the singleton-punctured oracle
$O_\star=O_{F_k^{\setminus\{x^\star\}}}$.
The first call computes the response into the clean scratch register
$\mathcal S$; the controlled phase
\(\Phi_{\bot}\lvert r,z\rangle
    =
    \omega^{\,r\mathbf 1[z=\bot]}
    \lvert r,z\rangle \)
adds phase $\omega^r$ exactly when the query address is the puncture.
The response is then added into $B$'s answer register, and the second
oracle call uncomputes $\mathcal S$.  After $B$ outputs its $q+1$ candidate pairs,
$\mathrm{Check}_{q+1}$ queries all reported inputs,
accepts only when exactly one response is $\bot$ and all other claimed
pairs verify, and outputs the pair already claimed at the punctured
input.  The total number of punctured-oracle calls is $2q+(q+1)=3q+1$.
}
\label{fig:singleton-extraction}
\end{figure}

%% file: implement.tex
We now show that the continuous phases used in Theorem~\ref{thm:singleton-puncture-smoothing} can be implemented with negligible error using a polynomial sized circuit.  We first define the function $\zeta(t,u)$ and show that it chooses the exact phases required for Theorem~\ref{thm:singleton-puncture-smoothing}. For $0\leq t < 1$ and $|u| = 1$, $\zeta(t,u)$ is defined as
    \[
    \zeta(t,u) = \frac{t+u}{1+tu}.
    \]
Note that since $|t+u|=|1+tu|$, we have $|\zeta(t,u)| = 1$.

\begin{lemma} \label{lem:poisson-phase}
    If $u$ is uniformly chosen on the unit circle, then for every $j\in \mathbb{Z}$,
    \[ \E[\zeta(t,u)^j] = t^{|j|}\]
\end{lemma}
\begin{proof}
    For $j\geq 0$ and $0\leq t<1$, 
    \[f_{j,t}(u) = \zeta(t,u)^j\] 
    is differentiable and continuous on all $u$ within the complex unit circle. Let $u = e^{i\theta}$ where $\theta\in [0,2\pi]$ is drawn uniformly at random, then the boundary mean-value property gives
        \[
            \frac{1}{2\pi}\int^{2\pi}_0 f_{j,t}(e^{i\theta})^j d\theta = f_{j,t}(0) = t^j.
        \]
         The same holds true for $j< 0$, since $|\zeta(t,u)| = 1$, and that $\zeta(t,u)^{-j}$ is the complex conjugate of $\zeta(t,u)^{j}$.
\end{proof}

\paragraph{Sampling $u_x$ efficiently.} We want to generate phases $\zeta(t,u)$ for each message $x$ independently at random. Let $a_x$ be the distinct integer encoding for each message $x$ where $d$ is the number of such distinct messages. Now independently at random sample a vector of seeds, $\mathbf{s}$, such that 
\[
\mathbf{s} = (s_0,..,s_{d-1}) 
\in \mathbb{T}^{d}\,\, ;
\qquad 
\mathbb{T} = \mathbb{R}/\mathbb{Z}.
\]
Then for each message $x$ define,
\[
\theta_x = \sum_{j=0}^{d-1} a_x^j\cdot s_j \quad (\text{mod }1) \, .
\]
Then we set $u_x = e^{2\pi i \theta_x}$ and the phases $z_x:= \zeta(t,u_x)$ that are required for Theorem~\ref{thm:singleton-puncture-smoothing}. 

\begin{lemma}[$d$-wise independence] \label{lem:phase-independence}
    For any $r\leq d$, every collection of at most $r$ distinct variables $\theta_x$ is jointly independently random over $\mathbb{T}^r$.
\end{lemma}
\begin{proof}
    Let $x_1,..x_r$ be $r\leq d$ distinct messages. The map $V: \mathbb{T}^d \rightarrow \mathbb{T}^r$  is defined as 
    \[
    V_{ij} = a^j_{x,i} 
    \qquad
    i \in \{1,..r\} , 
    \quad
    j \in \{0,..d-1\}
    \]
    which maps the seed vector $(s_0,..s_{d-1}) \mapsto (\theta_{x,1},..\theta_{x,r})$. 
    For any $i,k\leq r$ and $i\neq k$, let $\theta_{x,i}$ and $\theta_{x,k}$ be two elements of the output vector $\theta_x$. We have
    \[
        \theta_{x,i} - \theta_{x,k} = \sum_{j=0}^{j=d-1}(a_{x,i}^j-a_{x,k}^j)s_j \quad (\text{mod } 1).
    \]
    By definition, $a_{x,i}\neq a_{x,k}$ as each message has a unique integer encoding.
    Hence, $\theta_{x,i} - \theta_{x,k}$ is always non-zero. Thus, $\theta_x$ is linearly independent over all $\theta_{x,i}$ for all $i\leq r$. Thus the map $V$ takes the seed vector to a linearly independent vector. Since the seed vector is independently uniform at random, the output vector $\theta_x$ from the map $V$, remains independently uniform at random. 
\end{proof}
With Lemma~\ref{lem:poisson-phase} and Lemma~\ref{lem:phase-independence}, for any $
d$ distinct messages $x$, the phase can be implemented as 
\[
    |x\rangle 
    \mapsto 
    \zeta(t,u_x) \
    |x\rangle,
    \quad
    u_x = e^{2\pi i \theta_x},
\] 
which is the exact phase required for Theorem~\ref{thm:singleton-puncture-smoothing}.

\input{approx-chain-ckt}

Now it remains to show that when implementing $\zeta(t,u_x)$ both $t$ and $u_x$ can be sampled independently at random for at most $2q+1$ messages efficiently with up to negligible error. 
\begin{theorem}[Efficient Implementation] \label{thm:efficient-implementation} Let $A$ be a $q$-query plus-one adversary with success rate $\delta$. Then for $\widetilde{z} = (\zeta(\widetilde{t},\widetilde{u_x}))_{x\in \X}$ where    $\widetilde{t}$ is discrete and uniform on $[0,1-\tau]$ and for $x\in\X$, $\widetilde{u_x}$ is discrete and uniform on the unit circle, we have 
    \[
    \bigl|
        \mathbb E_{t,z} \Xi(A_{t,z})
        -
        \mathbb E_{\widetilde{t}, \widetilde{z}} \Xi(\widetilde{A_{t,z}})
    \bigr|
    \le \negl(\tau).
    \]
\end{theorem}

\begin{proof}
Note that, when $t=1$, the denominator of $\zeta(t,u_x)$ may be $0$. 
We thus want to discard an interval $\tau$ near $t=1$ to ensure an efficient implementation of $\zeta(t,u_x)$. 
Thus for any message $x$, $0\leq t\leq 1-\tau$ and $|u_x|=1$,
\[
    |1+tu_x|\geq 1-t\geq \tau.
\]
Let $\Xi(A)^{\tau}$ be the success measurement for when $t$ is sampled uniformly from $[0, 1-\tau]$. The distance between the success measurement in the ideal case $\Xi(A)$ and $\Xi(A)^{\tau}$ can be bounded for $q+1$ output messages as
\[
    \left|\E \Xi(A) - \E \Xi(A)^{\tau}\right| \leq (q+1)\tau.
\]

Due to the truncated domain of $t$, for any $|u_x|=1$ and in $t\leq 1-\tau$, we can bound the partial derivatives of function $\zeta(t,u_x)$ as
\begin{align}
    \frac{\partial\zeta(t,u_x)}{\partial u_x} &= \frac{1-t^2}{|1+tu_x|^2} \leq \frac{2}{\tau} \\
    \frac{\partial\zeta(t,u_x)}{\partial t} &= \frac{|1-u_x^2|}{|1+tu_x|^2} \leq \frac{2}{\tau^2}.
\end{align}

Let $\widetilde{t}$ and $\widetilde{u_x}$ be a discrete uniform sampling of $t$ and $u_x$ for each message $x$. Set
\[
    |t - \widetilde{t}| \leq \eps_{t}, 
    \qquad
    \sup_{x}|u_x - \widetilde{u_x}| \leq \eps_{u}.
\]
The approximation of $\zeta(t,u_x)$ over the truncated domain, on the discrete inputs $\widetilde{t}$ and $\widetilde{u_x}$ is given by
\[
    \zeta(t,u_x)\Big|_{\widetilde{t},\widetilde{u_x}} = 
    \frac{\partial \zeta(t,u_x)}{\partial t} |t-\widetilde{t}| +
    \frac{\partial \zeta(t,u_x)}{\partial u_x} |u_x-\widetilde{u_x}|.
\]
We have the phase for message $x$ as, $z_x = \zeta(t,u_x)$ and its discrete implementation as  $\widetilde{z_x} = \zeta(\widetilde{t},\widetilde{u_x})$. For $t\leq 1-\tau$ we can bound the change in the phase due to discretizations as 
\begin{align*} 
    \sup_{x}|z_x - \widetilde{z_x}| &\leq \frac{2\eps_t}{\tau^2} + \frac{2\eps_{u}}{\tau} := \eps_1.
\end{align*}
As show in Figure~\ref{fig:finite-precision-smoothing}, we now want to evaluate how this error propagates through the algorithm.
Let $R_z$ be the phase operator implementing the required phase $z_x$ for all messages $x$ and let $\widetilde{R_z}$ be its discrete implementation. 
\[
R_z : |x\rangle \mapsto z_x |x\rangle ,
\qquad
\widetilde{R_z} : |x\rangle \mapsto \tilde z_x |x\rangle .
\]

We have 
\[
\|R_z - \widetilde{R_z}\|\leq \eps_1.
\]
The smoothing circuit applies $q$ of the phase operators. Furthermore, the final state has the product of phases associated with the $q+1$ output messages. Let $C_R$ and $\widetilde{C_R}$ be the complete smoothing circuits for the exact and discrete phases respectively. The standard hybrid argument gives 
\[
    \|C_R - \widetilde{C_R}\| \leq q\eps_1 + (q+1)\eps_1 = (2q+1)\eps_1.
\]
Now the output is restricted to histories that are zero-touched at any fixed output message $x$. If $\alpha$ and $\widetilde{\alpha}$ are the amplitudes of such zero-touch histories in the ideal and discrete case respectively, with $\|\alpha\|^2, \|\widetilde{\alpha}\|^2 \leq 1$ we have
\[
    \left|\|\alpha\|^2 - \|\widetilde{\alpha}\|^2 \right| \leq 2(2q+1)\eps_1.
\]
Summing over such $q+1$ output messages, the distance between the success measurements in the truncated setting $\Xi(A)^{\tau}$ and the discrete implementation $\Xi\widetilde{(A)^{\tau}}$ 
\[
    \left| \E_z \Xi(A)^{\tau} - \E_{\widetilde{z}} \Xi\widetilde{(A)^{\tau}} \right|
    \leq 2(q+1)(2q+1)\eps_1.
\]

Hence the total error accumulated due to truncation and discrete implementation is
\[
    \left| \E_z \Xi(A) - \E_{\widetilde{z}} \Xi\widetilde{(A)^{\tau}} \right|
    \leq (q+1)\tau +2(q+1)(2q+1)\eps_1.
\]

With $q=\poly(\secpar)$ we can set $\tau = 2^{-\poly(\secpar)}$. Furthermore, let the error due to the discrete implementation can be set as $\eps_t = O(\tau^3)$ and $\eps_u = O(\tau^2)$. This ensures that the total error propagated due to truncation and discrete implementation is at most negligible.
\end{proof}

%% file: approx-chain-ckt.tex
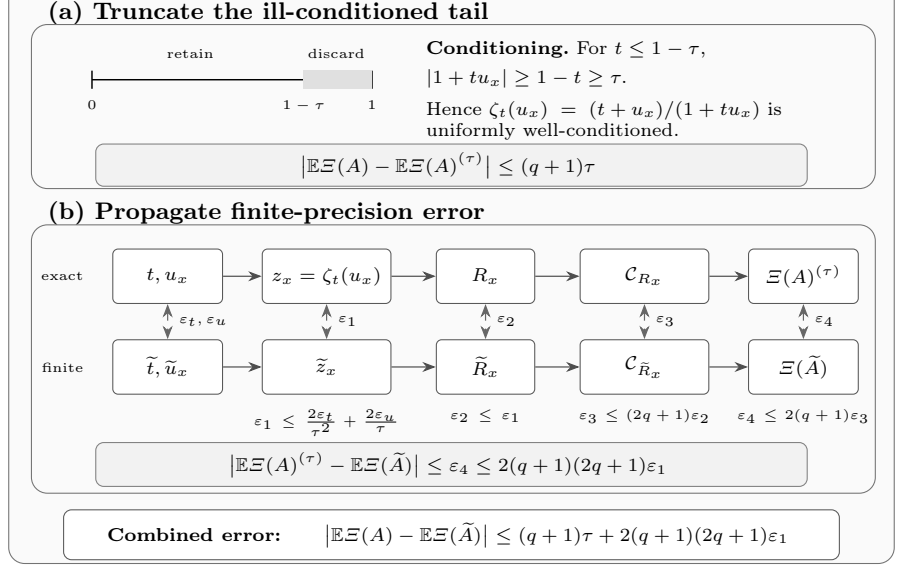
\begin{figure}[ht]
\centering
\begin{tikzpicture}[
  x=1cm,y=1cm,
  >={Stealth[length=1.8mm]},
  arrow/.style={
      ->,
      line width=0.45pt,
      draw=black!70
  },
  err/.style={
      <->,
      line width=0.45pt,
      draw=black!60,
      dashed
  },
  box/.style={
      draw=black!65,
      fill=white,
      rounded corners=2pt,
      minimum height=0.72cm,
      inner sep=3pt,
      align=center,
      font=\scriptsize
  },
  panel/.style={
      draw=black!55,
      fill=black!2,
      rounded corners=6pt
  },
  lab/.style={
      font=\scriptsize,
      align=center
  },
  tinylab/.style={
      font=\tiny,
      align=center
  }
]


\node[
  panel,
  fit={(-0.15,1.55) (11.55,-6.18)},
  inner sep=0pt
] (outer) {};

\node[
    draw=black!55,
    fill=black!55,
    rounded corners=3pt,
    anchor=west,
    font=\bfseries\footnotesize,
    text=white
]
at (0.10,1.55)
{Fig.~\ref{fig:finite-precision-smoothing}: Finite-precision Poisson smoothing};


\node[
    font=\bfseries\footnotesize,
    anchor=west
]
at (0.25,1.10)
{(a) Truncate the ill-conditioned tail};

\node[
    panel,
    fit={(0.15,0.95) (11.30,-1.22)},
    inner sep=0pt
] (topbox) {};

\draw[line width=0.55pt]
    (0.95,0.22) -- (4.65,0.22);

\draw[line width=0.55pt]
    (0.95,0.10) -- (0.95,0.34);

\draw[line width=0.55pt]
    (3.75,0.10) -- (3.75,0.34);

\draw[line width=0.55pt]
    (4.65,0.10) -- (4.65,0.34);

\fill[black!12]
    (3.75,0.10) rectangle (4.65,0.34);

\node[tinylab,anchor=north]
    at (0.95,0.06) {$0$};

\node[tinylab,anchor=north]
    at (3.75,0.06) {$1-\tau$};

\node[tinylab,anchor=north]
    at (4.65,0.06) {$1$};

\node[tinylab,anchor=south]
    at (2.25,0.39) {retain};

\node[tinylab,anchor=south]
    at (4.18,0.39) {discard};


\node[
    font=\scriptsize,
    text width=5.25cm,
    align=left,
    anchor=north west
]
at (5.25,0.85)
{
    \textbf{Conditioning.}
    For $t\le 1-\tau$,
    \\[3pt]
    $\displaystyle
        |1+t u_x|
        \ge
        1-t
        \ge
        \tau .
    $
    \\[4pt]
    Hence
    $\zeta_t(u_x)
        = {(t+u_x)}/{(1+t u_x)}
    $
    is uniformly well-conditioned.
};


\node[
    draw=black!50,
    fill=black!5,
    rounded corners=3pt,
    font=\scriptsize,
    align=center,
    minimum width=9.40cm,
    minimum height=0.52cm
]
at (5.70,-0.91)
{
    $\displaystyle
    \bigl|
        \mathbb E \Xi(A)
        -
        \mathbb E \Xi(A)^{(\tau)}
    \bigr|
    \le
    (q+1)\tau
    $
};


\node[
    font=\bfseries\footnotesize,
    anchor=west
]
at (0.25,-1.55)
{(b) Propagate finite-precision error};

\node[
    panel,
    fit={(0.15,-1.70) (11.30,-5.25)},
    inner sep=0pt
] (botbox) {};

\node[tinylab,anchor=east]
    at (0.95,-2.38)
    {exact};

\node[tinylab,anchor=east]
    at (0.95,-3.58)
    {finite};

\def\yT{-2.38}
\def\yB{-3.58}

\def\xA{1.95}
\def\xB{4.05}
\def\xC{6.15}
\def\xD{8.25}
\def\xE{10.35}


\node[
    box,
    minimum width=1.45cm
] (A1)
at (\xA,\yT)
{$t,u_x$};

\node[
    box,
    minimum width=1.70cm
] (A2)
at (\xB,\yT)
{$z_x=\zeta_t(u_x)$};

\node[
    box,
    minimum width=1.30cm
] (A3)
at (\xC,\yT)
{$R_x$};

\node[
    box,
    minimum width=1.70cm
] (A4)
at (\xD,\yT)
{$\mathcal C_{R_x}$};

\node[
    box,
    minimum width=1.45cm
] (A5)
at (\xE,\yT)
{$\Xi(A)^{(\tau)}$};


\node[
    box,
    minimum width=1.45cm
] (B1)
at (\xA,\yB)
{$\widetilde t,\widetilde u_x$};

\node[
    box,
    minimum width=1.70cm
] (B2)
at (\xB,\yB)
{$\widetilde z_x$};

\node[
    box,
    minimum width=1.30cm
] (B3)
at (\xC,\yB)
{$\widetilde R_x$};

\node[
    box,
    minimum width=1.70cm
] (B4)
at (\xD,\yB)
{$\mathcal C_{\widetilde R_x}$};

\node[
    box,
    minimum width=1.45cm
] (B5)
at (\xE,\yB)
{$\Xi(\widetilde A)$};


\draw[arrow] (A1.east) -- (A2.west);
\draw[arrow] (A2.east) -- (A3.west);
\draw[arrow] (A3.east) -- (A4.west);
\draw[arrow] (A4.east) -- (A5.west);


\draw[arrow] (B1.east) -- (B2.west);
\draw[arrow] (B2.east) -- (B3.west);
\draw[arrow] (B3.east) -- (B4.west);
\draw[arrow] (B4.east) -- (B5.west);


\draw[err]
    (A1.south) -- (B1.north)
    node[midway,right=1pt,tinylab]
    {$\varepsilon_t, \varepsilon_u$};

\draw[err]
    (A2.south) -- (B2.north)
    node[midway,right=1pt,tinylab]
    {$\varepsilon_1$};

\draw[err]
    (A3.south) -- (B3.north)
    node[midway,right=1pt,tinylab]
    {$\varepsilon_2$};

\draw[err]
    (A4.south) -- (B4.north)
    node[midway,right=1pt,tinylab]
    {$\varepsilon_3$};

\draw[err]
    (A5.south) -- (B5.north)
    node[midway,right=1pt,tinylab]
    {$\varepsilon_4$};


\node[
    tinylab,
    text width=2.45cm,
    anchor=north
]
at (\xB,-4.02)
{
    $\varepsilon_1
      \le
      \frac{2\varepsilon_t}{\tau^2}
      +\frac{2\varepsilon_u}{\tau}$
};

\node[
    tinylab,
    text width=1.50cm,
    anchor=north
]
at (\xC,-4.02)
{
    $\varepsilon_2
      \le
      \varepsilon_1$
};

\node[
    tinylab,
    text width=1.95cm,
    anchor=north
]
at (\xD,-4.02)
{
    $\varepsilon_3
      \le
      (2q+1)\varepsilon_2$
};

\node[
    tinylab,
    text width=2.00cm,
    anchor=north
]
at (\xE,-4.02)
{
    $\varepsilon_4
      \le
      2(q+1)\varepsilon_3$
};


\node[
    draw=black!50,
    fill=black!5,
    rounded corners=3pt,
    font=\scriptsize,
    align=center,
    minimum width=9.40cm,
    minimum height=0.58cm
]
at (5.70,-4.87)
{
    $\displaystyle
    \bigl|
        \mathbb E \Xi(A)^{(\tau)}
        -
        \mathbb E \Xi(\widetilde A)
    \bigr|
    \le
    \varepsilon_4
    \le
    2(q+1)(2q+1)\varepsilon_1
    $
};


\node[
    draw=black!65,
    fill=white,
    rounded corners=3pt,
    font=\scriptsize,
    align=center,
    minimum width=10.25cm,
    minimum height=0.66cm
]
at (5.70,-5.80)
{
    \textbf{Combined error:}\qquad
    $\displaystyle
    \bigl|
        \mathbb E \Xi(A)
        -
        \mathbb E \Xi(\widetilde A)
    \bigr|
    \le
    (q+1)\tau
    +
    2(q+1)(2q+1)\varepsilon_1
    $
};

\end{tikzpicture}
\caption{
Finite-precision implementation of Poisson smoothing.
\textbf{Top.}
Truncating the smoothing parameter to $t\in[0,1-\tau]$ loses at most
$(q+1)\tau$ in expected singleton zero-touch mass and guarantees
$|1+t u_x|\ge\tau$, making
$\zeta_t(u_x)=(t+u_x)/(1+t u_x)$ uniform and computable.
\textbf{Bottom.}
After truncation, finite-precision error is propagated through sampling
$(t,u_x)$, evaluation of the phase multipliers, implementation of the
phase operators, and substitution into the smoothed $q$-query circuit.
The term $\varepsilon_{t}, \varepsilon_{u}$
is the additional numerical error incurred
$\zeta({\widetilde t}, \widetilde u_x)$ at
input discretization error.
The remaining arrows propagate this phase error through the implemented
unitaries and the full smoothed circuit, and the local truncation and
finite-precision bounds combine to give the final approximation between
the ideal and implementable zero-touch masses.
}

\label{fig:finite-precision-smoothing}
\end{figure}

%% file: spdu2po.tex
\begin{corollary}[$\dBU$ implies $\PO$]
\label{cor:spdu-implies-po}
If a function family $F$ is $\dBU$-secure, then
$F$ is $\PO$-secure.  Specifically, any $q$-query $\PO$ adversary $A$ succeeding with
probability $\delta$ yields an efficient singleton-puncture adversary $C$
using at most $3q+1$ queries and satisfying
\[
    \rho_F(C)
    \ge
    \frac{\delta}{4(q+1)}
    -
    \negl(\secpar).
\]
\end{corollary}

\begin{proof}
Let $A$ be a $q$-query $\PO$ adversary with success probability $\delta$.
By the Poisson-smoothing (\cref{thm:singleton-puncture-smoothing}),
\[
    \E_{t,z}\,\Xi(A_{t,z})
    \ge
    \frac{\delta}{4(q+1)}.
\]
By Theorem~\ref{thm:efficient-implementation}, the ideal smoothing can be
replaced by its efficient finite-precision implementation with only
negligible loss.  Thus, if a $q$-query $\PO$ adversary $B$ samples
$(\widetilde t,\widetilde z)$ internally and runs the corresponding
smoothed adversary,
\[
    \Xi(B)
    =
    \E_{\widetilde t,\widetilde z}
        \Xi(\widetilde A_{\widetilde t,\widetilde z})
    \ge
    \frac{\delta}{4(q+1)}
    -
    \negl(\secpar).
\]

Applying the singleton zero-touch extraction (\cref{lemma:singlton-ext}) to $B$
produces an efficient adversary $C$ making at most $3q+1$ queries to the
singleton-punctured oracle and satisfying
\[
    \rho_F(C)
    \ge
    \Xi(B) \ge
    \frac{\delta}{4(q+1)}
    -
    \negl(\secpar) \, .
\]

Consequently, if $F$ is $\dBU$, then
$\rho_F(C)$ is negligible for every efficient $C$, which forces $\delta$
to be negligible.  Hence $F$ is $\PO$-secure.
\end{proof}

We then obtain the  $\BU$ to $\PO$ implication immediately, and the quantitative bound follows from~\cref{lemma:bu2dbu}. 

\begin{corollary}[$\BU$ implies $\PO$]
\label{cor:bu-implies-po}
Let $A$ be a $q$-query adversary that wins the $\PO$ game against
$F$ with probability $\delta$. Then there exists an efficient
$\BU$ adversary $C$, making at most $3q+1$ oracle queries, such that
\[
    \sup_{\beta\in[0,1]}
    \Pr\!\left[C\text{ wins }\BU_\beta\right]
    \ge
    \frac{\delta}
         {8(q+1)(6q+3)^2}
    -
    \negl(\secpar) = \Omega\!\left(\frac{\delta}{q^3}\right)
    -
    \negl(\secpar) \, .\]

Moreover, for the explicit blinding rate
\[
    \beta_0
    :=
    \frac{\delta}
         {16(q+1)(6q+3)^4},
\]
the same adversary satisfies
\[
    \Pr\!\left[C\text{ wins }\BU_{\beta_0}\right]
    \ge
    \frac{\delta^2}
         {128(q+1)^2(6q+3)^4}
    -
    \negl(\secpar).
\]

Consequently, $\BU$ security implies $\PO$ security.
\end{corollary}

%% file: counterexample.tex
From the discussion above, we know that for an attacker $A$ to win $\BU$, it has to have large singleton zero-touch mass $\Xi(A)$. Here we give an explicit example where a $\PO$ attacker $A$ that succeeds with probability approaching one but having $\Xi(A)$ identically zero. Hence a reduction that attempts to use a $\PO$ attacker without modifying its interference pattern is not guaranteed to win $\BU$. More generally, this rules out any argument whose quantitative guarantee factors through $\Xi(A)$. We may call such arguments \emph{raw} singleton extraction.   

\begin{theorem}[All-or-nothing cancellation for raw singleton extraction]
\label{thm:aon-cancellation}
For every $q\ge 1$, there exists a keyed function family and a $q$-query algorithm $A$ such that
\[
    \Pr[A \text{ wins } \PO]
    =
    \frac{q^2}{(q+2)^2},
\]

while the singleton zero-touch mass of $A$ is identically \emph{zero}.  In particular, there is no general lower bound of the form

\[
    \Xi(A)
    \ge
    \frac{\Pr[A\text{ wins }\PO]}{\poly(q)}
\]

for the zero-touch mass obtained by regrouping the histories of the original, unmodified algorithm.
\end{theorem}

\begin{proof}
Let \(X=\{1,\ldots,q+1\}\) and let $c\notin X$ be an additional input point in the domain \(\X\).  Consider the keyed function

\[
    F_K(x)=K
    \qquad
    \text{for every }x\in X\cup\{c\}.
\]

The algorithm uses a branch register with orthonormal basis \(
    \{\ket{0}\}\cup\{\ket{a}:a\in X\}\) to control the computation. It is initialized to

\[
    |s\rangle
    :=
    \frac{1}{\sqrt{q+2}}
    \left(
        \ket{0}+\sum_{a\in X}\ket{a}
    \right).
\]

Conditioned on branch $0$, the algorithm queries the $c$ on
all $q$ oracle calls.  Conditioned on branch $a\in X$, it queries each $x$
in $X\setminus{a}$ exactly once.  These branch-dependent query addresses can
be computed coherently from the branch register before each oracle call and
uncomputed afterwards, so this defines a valid $q$-query quantum algorithm.

Since $F_K$ is constant on $X\cup{c}$, every branch receives the same oracle
value $K$.  Hence, after the $q$ queries, all non-branch registers may be made
identical across branches.  In particular, the state can be written in product form as
\(  |s\rangle_B\otimes |\psi_K\rangle\), where $|\psi_K\rangle$ contains $K$ and is independent of the branch label.

At the end, the algorithm performs on the branch register the two-outcome measurement 

\[
    \bigl\{
        \ket{t}\bra{t},
        I-\ket{t}\bra{t}
    \bigr\} \quad \text{ with }  \ket{t}
    :=
    \frac{1}{\sqrt{q+2}}
    (
        \ket{0}-\sum_{a\in X}\ket{a}
    ) \, . 
\]

Then on the first outcome, it outputs the $q+1$ valid
pairs \(
    \{(x,K):x\in X\}\); and on the second outcome it outputs some fixed invalid tuple.  
    
Therefore its $\PO$ success probability is

\[
    |\langle t|s\rangle|^2
    =
    \left|
        \frac{1-(q+1)}{q+2}
    \right|^2
    =
    \frac{q^2}{(q+2)^2} \, .
\]

We now examine the raw singleton zero-touch amplitudes.  For the successful
output set $X$, let $v_M$ denote the coherent accepting amplitude contributed
by histories whose omitted reported set is exactly $M\subseteq X$.

\begin{itemize}
    \item  Branch $0$ queries none of the reported addresses and therefore has omitted
set $
    M_0=X$, and contributes to the accepting amplitude by
\(    v_X
    =
    \frac{1}{q+2}|\psi_K\rangle\).

    \item For each $a\in X$, branch $a$ queries all reported points except $a$, and
hence has omitted set $M_a=\{a\}$, and since $\langle t|a\rangle=-1/\sqrt{q+2}$, its contribution to the accepting amplitude is \( v_{\{a\}}
    = -\frac{1}{q+2}|\psi_K\rangle\).
\end{itemize}

Consequently, the total accepting amplitude is

\[
    v_X+\sum_{a\in X}v_{\{a\}}
    =
    -\frac{q}{q+2}\ket{\psi_K}\, ,
\]
as required by the success calculation above.

Fix now any $a\in X$.  Exactly two branches are zero-touch at $a$: branch $0$, which omits all of $X$, and branch $a$, which omits only $a$.  Hence the singleton zero-touch amplitude at $a$ is

\[
    u_a
    =
    v_X+v_{\{a\}}
    =
    \frac{1}{q+2}\ket{\psi_K}
    -
    \frac{1}{q+2}\ket{\psi_K}
    =
    0\, .
\]

Thus

\[
    \Xi(A) := \sum_{a\in X} \norm{u_a}^2 = 0 \, ,
\]

despite the fact that $A$ wins $\PO$ with probability
$q^2/(q+2)^2$.
\end{proof}

The example shows that the path-wise pigeonhole principle does not survive
naive coherent regrouping.  Every successful history omits at least one
reported point, but histories with different omitted-set cardinalities can
interfere so that every singleton zero-touch class vanishes.  Hence in order to extract a zero-touch witness, we
need to first \emph{reshape} the original interference pattern, and this is exactly what Poisson smoothing realizes.

For a parameter $t\in[0,1]$, Poisson smoothing has the effect that, after
averaging over its random phases, a successful history omitting $r$ reported
points acquires the scalar multiplier $t^{2r-1}$.

In the example above there are only two omission layers relevant to a fixed
reported point $a$:

\[
    |M|=q+1
    \qquad\text{and}\qquad
    |M|=1.
\]

Their respective multipliers are therefore \(t^{2q+1}\)  and \(t\).  Hence the smoothed zero-touch amplitude at $a$ becomes

\[
    u_a(t)
    =
    \frac{t^{2q+1}-t}{q+2}
    \ket{\psi_K} \, .
\]

We can see that at $t=1$ we recover the original exact cancellation \(u_a(1)=0\).  But for every $0<t<1$,
\(t^{2q+1}\neq t\),  and therefore

\[
    u_a(t)\neq 0 \, .
\]

Thus Poisson smoothing does not attempt to identify which history omitted
which reported point, nor does it isolate histories omitting exactly one
point.  Instead, it continuously changes the relative weights of different
omission-cardinality layers.  A cancellation that is perfectly tuned in the
original computation may occur at one value of $t$, but it cannot remain
small throughout an interpolation that moves from zero amplitude at $t=0$
to the intact successful amplitude at $t=1$.

%% file: appendix.tex
\section{Algebraic Property of \(\BU\) success and \(\rho_F\)}
\label{sec:polychar}
\begin{proposition}[\cref{prop:polychar} restated] 
Let \(A\) be a \(q\)-query algorithm in the \(\BU\) game. Define 
\[ G_{A,k} (\beta) : = \Pr[ A^{F_k^{{\setminus B}}} \text{ wins } \BU: B\xleftarrow{\beta}X_\secpar ] \, .\]
Then \( G_A(\beta):=\E_k\!\left[G_{A,k}(\beta)\right]\) is a polynomial in \(\beta\) of degree at most \(d = 2q+1 \), and moreover
\[ G_A(0)=0, \qquad
 G_A'(0)=\rho_F(A) \, .\]
\label{prop:polychar-restate}
\end{proposition}

\begin{proof}
Fix a key \(k\), and introduce Boolean variables
\[
\Vec{b}:=(b_x)_{x\in X_\secpar}\in\{0,1\}^{X_\secpar},
\]
where \(b_x=1\) iff \(x\) is blinded. Hence the blind set \(B (\Vec{b}) = \{x: b_x = 1\}\), and we can regard \(F_k^{\setminus B(\Vec{b})}\) as an oracle parameterized by \(\Vec{b}\).

Recall the standard unitary oracle for \(F_k\) is specified by
\[ U_{F_k}|x,y,z\rangle =
|x,y\oplus F_k(x),z\rangle \, . \]
For the blinded oracle, encode \(\bot\) as an additional response value and define 
\[U_{\bot} |x,y,z\rangle =
|x,y\oplus \bot,z\rangle \, .\]

Then the blinded unitary oracle for \(F_k^{\setminus B(\Vec{b})}\) can be described as 
\[O_{k,\Vec{b}} = \sum_x |x\rangle\langle x|
\otimes
\Big((1-b_x)U_{F_k(x)}+ b_xU_\bot\Big) \, . \]

Therefore, we observe that for each query \(O_{k,\Vec{b}}\) acts as an affine transformation in
the variables \(b_x\).  Consequently, after \(q\) queries, every final
amplitude of \(A\) is a polynomial in \(\Vec{b}\) of total degree at most \(q\).
Since \(\{b_x\}\) are Boolean variables, we can multi-linearize these
polynomials using \(b_x^r=b_x (r\ge 1)\), without changing their values or increasing the degree. Write the final state of \(A\) as 
\[
|\psi_{k,\Vec{b}}\rangle
=
\sum_{x\in X_\secpar}\sum_{y\in Y_\secpar}
\sum_{w}
\alpha_{k,x,y,w}(\Vec{b})\,|x,y,w\rangle,
\]
where each
\(\alpha_{k,x,y,w}(b)\) is a multilinear polynomial of degree at most
\(q\).

For each possible classical output \((x,y)\), define
\[
p_{k,x,y}(\Vec{b})
:=
\sum_w
\left|\alpha_{k,x,y,w}(\Vec{b})\right|^2, 
\]
which denotes exactly the probability that \(A\)
outputs \((x,y)\).  Since it is a sum of products
\[
\alpha_{k,x,y,w}(\Vec{b})\,
\overline{\alpha_{k,x,y,w}(\Vec{b})},
\]
it agrees with a multilinear polynomial of degree
at most $2q$.

Since an output $(x,y)$ is a winning BU
forgery exactly when
\[
y=F_k(x)
\qquad\text{and}\qquad
b_x=1 \, ,
\]
the BU success probability for fixed \(k\) and \(\Vec{b}\) is
\[
P_k(\Vec{b})
:= \Pr\!\left[
A^{F_k^{\setminus B(b)}}\text{ wins }\BU
\right] = 
\sum_{x\in X_\secpar}
b_x\,p_{k,x,F_k(x)}(\Vec{b}) \, ,
\]
which is a multilinear polynomial of degree in \(\Vec{b}\) at most
\(2q+1\). Let us write the multilinear expansion
\[
P_k(\Vec{b})
=
\sum_{\substack{S\subseteq X_\secpar\\|S|\le 2q+1}}
c_{k,S}\prod_{x\in S}b_x \, ,
\]
with coefficients \(c_{k,S}\). 

Now let the coordinates $b_x$ be independent
$\operatorname{Bernoulli}(\beta)$ random variables.  By definition,
\[
G_{A,k}(\beta)
=
\E_{\Vec{b}\sim\operatorname{Ber}(\beta)^{X_\secpar}}
[P_k(\Vec{b})] =
\sum_{\substack{S\subseteq X_\kappa\\|S|\le 2q+1}}
c_{k,S}\beta^{|S|}
=
P_k(\beta,\ldots,\beta) \, .
\]

This becomes a univariate polynomial in \(\beta\) of degree at most $2q+1$.  Averaging over the random key \(k\) preserves the degree bound, so \(G_A(\beta):=\E_k\!\left[G_{A,k}(\beta)\right] \) has degree at most \(2q+1\). 

If \(\beta=0\), then \(B=\varnothing\) with probability one.  There is no
blinded point on which $A$ can forge, and therefore \(G_A(0)=0\). It remains to identify the derivative at the origin.  For fixed $k$,
let
\[
h_k(S)
:=
\Pr\!\left[
A^{F_k^{\setminus S}}\text{ wins }\BU
\right]
\]
when the blind set is exactly $S$.  Then
\[
G_{A,k}(\beta)
=
\sum_{S\subseteq X_\kappa}
h_k(S)\,
\beta^{|S|}(1-\beta)^{|X_\secpar|-|S|}.
\]
Since \(h_k(\varnothing)=0\), differentiating at $\beta=0$ leaves only
the singleton sets:
\[
G_{A,k}'(0)
=
\sum_{a\in X_\secpar} h_k(\{a\}).
\]
For \(S=\{a\}\), winning the BU game means precisely outputting the
missing pair \((a,F_k(a))\).  Thus
\[
h_k(\{a\})
=
\Pr\!\left[
A^{F_k^{\setminus\{a\}}}
\text{ outputs }(a,F_k(a))
\right].
\]
Averaging over $k$ gives
\[
G_A'(0)
=
\E_k\sum_{a\in X_\kappa}
\Pr\!\left[
A^{F_k^{\setminus\{a\}}}
\text{ outputs }(a,F_k(a))
\right]
=
\rho_F(A).
\]
\end{proof}